\documentclass[superscriptaddress]{revtex4-2}

\usepackage{amssymb,amsmath,amsfonts,amssymb}
\usepackage{graphics,graphicx,color, subcaption}
\usepackage{empheq}

\def\@abssec#1{\vspace{.05in}\footnotesize \parindent .2in
{\bf #1. }\ignorespaces}
\def\proof{\par{\it Proof}. \ignorespaces}
\def\endproof{{\ \vbox{\hrule\hbox{%
   \vrule height1.3ex\hskip0.8ex\vrule}\hrule
  }}\par}
\graphicspath{{/EPSF/}{Figures/}}
\DeclareGraphicsExtensions{.eps}
\usepackage[margin=0.85in]{geometry}
\usepackage{tikz-cd}

\newtheorem{theorem}{Theorem}[section]

\def \Rm {\mathbb R}

\def \Cm {\mathbb C}

\def \Sm {\mathbb S}

\newcommand{\dsum}{\displaystyle\sum}

\newcommand{\pdr}[2]{\dfrac{\partial{#1}}{\partial{#2}}}

\newcommand{\bk}{\mathbf k}

\newcommand{\mC}{\mathcal C}
\newcommand{\mF}{\mathcal F}
\newcommand{\mH}{\mathcal H}

\newcommand{\mJ}{{\mathfrak J}}

\newcommand{\fC}{{\mathfrak C}}

\newcommand{\cout}[1]{}

\newcommand{\op}{\omega_p}
\newcommand{\om}{\Omega}

\newcommand{\kp}{{\mathbf{k}}} %% k_\perp
\newcommand{\BDI}{{\rm BDI}}

\begin{document}
%%%%%%%%%%%%%
%%% TITLE %%%
%%%%%%%%%%%%%
\title{Effects of interface regularity on the bulk-edge correspondence in continuum photonic systems}

%%%%%%%%%%%%%%%
%%% AUTHORS %%%
%%%%%%%%%%%%%%%
\author{Matthew Frazier}
\email{Corresponding author: mjfrazier@uchicago.edu}
\affiliation{%
 Committee on Computational and Applied Mathematics, University of Chicago, Chicago, IL 60637, USA
}%

\author{Guillaume Bal}
 \affiliation{%
 Committee on Computational and Applied Mathematics, University of Chicago, Chicago, IL 60637, USA
}
\affiliation{%
 Departments of Mathematics and Statistics, University of Chicago, Chicago, IL 60637, USA
}%

\begin{abstract}
In this study we analyze the topological invariants and edge states of transverse magnetic wave propagation in continuum photonic systems at a finite-width interface between two gyrotropic materials with different magnetic biases. Where previous studies have almost exclusively considered sharp transitions between two different electromagnetic media, we consider the more general geometry where the magnetic field bias is allowed to vary arbitrarily in a finite-width interface between two bulk regions. We find that when the magnetic field bias varies continuously between the two bulk regions, the Bulk Edge Correspondence (BEC) holds robustly with respect to well-defined Chern invariants. However, discontinuities in the magnetic field bias introduce new edge modes and anomalous high-wavenumber spectral effects which alter the BEC. We analyze these spectral alterations and define a new anomalous BEC in continuum photonic systems which includes contributions from topological invariants and discontinuities in magnetic field bias.
\end{abstract}
\maketitle

\section{Introduction}
The theory of topological insulators, originally developed for quantum hall systems \cite{kane2005quantum, bernevig2013topological, shen2012topological} and successfully applied to geophysical shallow water models \cite{delplace2017topological, tauber2019bulk, bal2022topological}, 2d Moir\'e materials \cite{bal2023mathematical, bal2022multiscale}, layered graphene \cite{haldane1988model, min2008electronic, zhang2011spontaneous}, and periodic photonic systems \cite{haldane2008possible, raghu2008analogs, wang2008reflection}, provides a characterization of robust asymmetric transport along edges between two insulating media via topological invariants. These invariants split the system into topological phases in parameter space, and the canonical relation termed the Bulk-edge Correspondence (BEC) states that the number of edge states appearing in an insulating gap at an interface between two topological phases is quantized by the difference in topological invariants across that interface. Although the BEC has been proved in quantum Hall systems \cite{hasan2010colloquium}, in limited cases for atmospheric equatorial waves \cite{tauber2019bulk, bal2022topological}, and for more general theoretical cases under distinct invariants \cite{delplace2022berry, fu2021topological}, a general proof for continuum systems is elusive and has proved incorrect for equatorial waves and continuum photonic systems under certain conditions \cite{gangaraj2018coupled, hassani2020physical, bal2024topological}. 

Continuum photonic systems, specifically Transverse Magnetic (TM) wave propagation in gyrotropic media, contain several barriers to topological classification and characterization of edge states via the BEC \cite{silveirinha2015chern, hassani2019truly, buddhiraju2020absence}. First, for the local Drude model of electromagnetic transport in this system, integrals of Berry curvature (which produce Chern invariants under appropriate conditions) of the system are not only non-integer in some cases, but can continuously vary with the parameters of the system \cite{silveirinha2015chern, hanson2016notes, hassani2016effects}. Silveirinha, by adding a physically motivated spatial cutoff (SC) regularization at high wavenumbers, was able to restore integer-valued integrals of Berry curvature for the system, which allowed the BEC to be evaluated in a meaningful way \cite{silveirinha2015chern, silveirinha2016bulk}. However, comparison of these invariants with the number of edge modes appearing in spectral calculations across a number of studies showed inconsistencies in the BEC both in relation to the total number of edge states differing from that predicted by topological invariants, and the appearance of a number of edge states which are consistent with the BEC but do not fully span the insulating gap \cite{gangaraj2018coupled, han2022anomalous, hassani2020physical,  hassani2019truly}. Further, we recently showed that although the SC regularization produces integer integrals of Berry curvature, it does so in spite of the fact that these invariants are not always well-defined Chern numbers \cite{frazier2025topological}. Application of an alternate regularization produced \textit{bona fide} Chern invariants which were distinct from the previous values in some cases, and we were able to show numerically that these new invariants satisfied the BEC in all cases where the transition between two topological phases was continuous. 

Continuity of the interface is an aspect which has played an important role in the BEC for geophysical shallow water models governing equatorial waves. The BEC was rigorously proved in this system for any case in which the derivative of the Coriolis force is bounded \cite{tauber2019bulk, bal2022topological, souslov2019topological}. In any case where a jump in the Coriolis force is present, the number of edge states appearing in a given energy range within the insulating gap is increased or decreased, respectively, by the number of positive or negative jumps in the Coriolis force. In addition, the (flat) asymptotic dispersion of certain edge modes, induced by discontinuities in the Coriolis force, is determined by the size of the jump itself so that edge modes do not always span the entire insulating gap \cite{bal2024topological}. In this way, the BEC is altered in a predictable way by the number and amplitude of discontinuities in the Coriolis force parameter.

The main result of this paper is that the role of continuity of the interface in the BEC for continuum photonic systems is nearly identical to that in equatorial waves. We derive here an anomalous BEC for continuum photonic systems which includes contributions to spectral flow from well-defined topological invariants and from anomalous high-wavenumber spectral effects which appear as a consequence of magnetic field jumps. Many studies which consider sharp interfaces between biased plasmas or between a biased plasma and a topologically trivial medium have indeed shown incomplete band gap coverage by edge states with flat asymptotic dispersion, similar to the equatorial wave case \cite{gangaraj2018coupled, hassani2020physical, hassani2019truly,  buddhiraju2020absence, han2022anomalous}. As a consequence they identify that the spectral flow in some band gaps is not robustly quantized by bulk topological invariants (and hence violates the BEC). As an explanation of the failure of the BEC most studies analyze the realistic effects of dissipation and the addition of non-local dynamics, which restores full band gap coverage by all edge modes and thus restores the BEC as a physical principle \cite{hassani2020physical, hassani2019truly,buddhiraju2020absence}. However, whether or not the spectral flow is influenced by interface effects, it still displays a remarkable quantized asymmetric edge transport spatially confined to the interface in the appropriate energy range, which under the correct conditions will be robust against scattering and dispersion into the bulk. Thus, a quantitative description of interface effects which allows us to utilize simpler models wherever they are physically relevant may be critically important for the design and implementation of continuum photonic systems where these remarkable properties are desired. By characterizing anomalies in the BEC in terms of the regularity of the interface, we thus aim to restore the BEC in continuum photonics as a mathematical principle and a powerful theoretical predictor of robust asymmetric transport properties. 

Finally, we note that some of the inconsistencies in the BEC for photonic continua were addressed previously by applying a non-local pressure term to the system, whose linearization provides a hydrodynamic model for the TM photonic system \cite{pakniyat2022chern, serra2025influence}. This model has integer-valued integrals of Berry curvature without the need for regularization, and they are in fact well-defined Chern numbers as we show in Section \ref{sec:hydro}. As such, we find that the BEC also holds for this model when the interface is continuous and similar anomalous edge states effects appear due to jumps in magnetic field. 

In this paper, we consider propagation of TM waves at the interface between two cold plasmas, modeling macroscopic electromagnetic oscillations in isotropic metals and semi-metals, with differing magnetic biases. We choose to focus on this particular interface rather than an interface between a biased plasma and isotropic opaque or transparent media in order to avoid additional complications from specifying additional boundary conditions \cite{silveirinha2016bulk} and focus on the effects of continuity of the magnetic field on the BEC. In order to evaluate the BEC we construct Bulk Difference Invariants (BDIs) \cite{bal2022topological, bal2019continuous, bal2023topological}, which provide a more flexible means to define \textit{bona fide} Chern numbers with which to evaluate the BEC than differences in bulk invariants. Construction of well-defined BDIs for a local Drude model requires high-wavenumber regularization \cite{silveirinha2015chern, frazier2025topological} and we analyze the BEC for the SC regularization, the BDI regularization we introduced in \cite{frazier2025topological}, and for well-defined Chern numbers applied to the unregularized system. In addition, we analyze the linearized hydrodynamic Drude model as introduced in \cite{pakniyat2022chern} as a minimal model of non-local effects. The remainder of the paper is organized as follows. In Section \ref{sec:background} we summarize the definition of BDIs and their relation to the BEC, introduce relevant results from the analysis of equatorial waves, and state our main result, which precisely modifies the BEC in continuum photonic systems based on the number and amplitude of magnetic field discontinuities. Sections \ref{sec:cold_plasma_model} and \ref{sec:reg_models} discuss the main results in detail as they apply to the local model and its regularized counterparts and provide numerical spectral calculations for further justification. Section \ref{sec:hydro} extends these results to the hydrodynamic model of TM wave propagation. We conclude with Section \ref{sec:discuss} and additional calculations and details can be found in the appendices. 

\section{Background}\label{sec:background}

This section summarizes the construction and computation of bulk and interface invariants and their application to the cold plasma model for photonic continua, which will be essential to our analysis in the following sections. We also summarize applicable results from the closely-related shallow water model for atmospheric dynamics. For additional details, see \cite{bernevig2013topological, bal2024topological, silveirinha2015chern, prodan2016bulk,bal2024continuous}.

\subsection{Bulk-difference invariants and bulk-edge correspondence}
\paragraph{BDIs.}First, consider a family of self-adjoint Hamiltonians in momentum-space $\hat H(\kp)$ for $\kp\in\Rm^2$ with values in $\Cm^{n\times n}$, and assume the following spectral decomposition:
\[  \hat H(\kp) = \dsum_{j=1}^n \lambda_j(\kp) \Pi_j(\kp).\]
Here $\Pi_j(\kp)=|\psi_j(\kp)\rangle\langle\psi_j(\kp)|$ are rank-one projectors and $\lambda_j(\bk)$ are their corresponding (real-valued) eigenvalues. Bulk Hamiltonians are assumed to have constant coefficients so that each momentum-space Hamiltonian $\hat H(\bk)$ is associated to a real-space Hamiltonian $H(D_x, D_y) = \mF^{-1}\hat H(\bk)\mF$ by the replacement $k_x\to D_x:= -i\partial_x$, $k_y \to D_y:= -i\partial_y$. Associated with each projector family $\Rm^2\ni \kp\mapsto\Pi_j(\kp)$ is the following integral of Berry curvature
\begin{equation}\label{eq:curv}
    \mC[\Pi_j]=\frac i{2\pi} \int_{\Rm^2} {\rm tr} \Pi_j d\Pi_j \wedge d\Pi_j = \frac 1{2\pi}\int_{\Rm^2}dA_j(\bk),\qquad 
d\Pi := \pdr{\Pi}{k_x}dk_x + \pdr{\Pi}{k_y}dk_y,
\end{equation}
for the Berry connection $A_j(\bk) = i\left(\psi_j(\bk), d\psi_j(\bk)\right)$. Here, ${\rm tr}$ stands for the standard matrix trace. Note that this definition can be extended to any projector of arbitrary-rank $\Pi$ using the first equality in \eqref{eq:curv}. We also assume a global band gap between bands $\ell$ and $\ell+1$, i.e., a frequency interval $I_\ell$ such that $\lambda_j(\kp)< I_\ell$ for $j\leq \ell$ while $\lambda_j(\kp)> I_\ell$ for $j \geq \ell+1$. Then denoting $P_\ell(\bk) = \sum_{j \le \ell} \Pi_j(\bk)$ we may define a bulk invariant which characterizes gap $\ell$ by $\mC[P_\ell]$, equal to $\sum_{j = 1}^\ell \mC[\Pi_j]$ by the additivity of Chern numbers \cite{bernevig2013topological}. When the domain of integration in \eqref{eq:curv} is a compact Brillouin zone then this integral is in fact a well-defined Chern number as the integral of the curvature of a 1-form Berry connection over a closed manifold (the Brillouin torus). In continuum systems we may attempt similarly to identify $\Rm^2$ with the closed manifold $\Sm^2$ by stereographic projection \cite{bal2022topological, silveirinha2015chern}, however we must verify the extra condition that the projector $\Pi$ is continuous (modulo a gauge transformation) at the point in $\Sm^2$ identified with $\bk \to \infty$. In many continuum systems, including the cold plasma model we consider in this paper, this criterion is not met without the addition of regularizing terms \cite{tauber2019bulk, bal2022topological, frazier2025topological}.

One way to overcome this difficulty is to define bulk-difference invariants (BDIs) which instead characterize the common band gap of two bulks which meet along a given interface. Given two bulk systems which are governed by the two Hamiltonians $\hat H^S(\bk)$ and $\hat H^N(\bk)$ respectively, we define the bulk-difference invariant for a spectral gap between bands $\ell$ and $\ell+1$ (labeled by $\ell$) shared by both $\hat H^S(\bk)$ and $\hat H^N(\bk)$ to be:
\begin{equation}\label{eq:ChernBDI}
    \fC_\ell := \fC[P_\ell^S,P_\ell^N] := \mC[P_\ell^S]-\mC[P_\ell^N]  = \frac  i{2\pi} \int_{\Rm^2} {\rm tr} P_\ell^S dP_\ell^S\wedge dP_\ell^S
 - \frac  i{2\pi} \int_{\Rm^2} {\rm tr} P_\ell^N dP_\ell^N \wedge dP_\ell^N.
\end{equation}

Provided that we have the following gluing condition 
\begin{equation}\label{eq:gluing}
    \lim_{k\to\infty} P_\ell^N(k,\theta) = \lim_{k\to\infty} P_\ell^S(k,\theta) \quad \mbox{ for all } \theta\in [0, 2\pi),
\end{equation}
$\fC_\ell$ is a {\em bona fide} Chern number for the family of projectors $\{P_\ell^N(\bk), P_\ell^S(\bk)\}$ defined continuously on the sphere $\Sm^2$, where $P_\ell^N$ is defined to be zero on the lower half of $\Sm^2$ and $P_\ell^S$ is defined to be zero on the upper half. Here $(k, \theta)$ corresponds to the polar coordinates of $\bk$ in $\Rm^2$. See \cite{bal2019continuous} for more details on the definition of BDIs and \cite{frazier2025topological} for an extensive analysis of the regularization conditions necessary to define BDIs for the cold plasma model. We will extend our analysis of BDIs in photonic continua to the non-local hydrodynamic model in Section \ref{sec:hydro}. Note that \eqref{eq:ChernBDI} corresponds to the usual difference in bulk invariants $\mC[P^S_\ell]-\mC[P^N_\ell]$, except $\mC[P^{N/S}_\ell]$ need not be integer-valued in order for $\fC_\ell$ to be a well-defined Chern number- only the extra gluing condition \eqref{eq:gluing} must be verified. 

\paragraph{BEC.} Subsequently, $\fC_\ell$ may be related to the (quantized) asymmetric edge current along the interface between the two bulk systems characterized by $H^S$ and $H^N$. Assume that $H^S$ and $H^N$ have a common insulating band gap $g = [E_0, E_1]$. Then we define the expectation of signal transport in the energy region $g$ along the interface by the quantity $\sigma_I(g)$. See Appendix \ref{sec:sigma_I} for a detailed discussion and definition of $\sigma_I$. In particular, $\sigma_I$ is defined via spectral calculus on the \textit{interface Hamiltonian} $H_I(x, y, D_x, D_y)$, for which $H_I = H_I(D_x, D_y)= H^N$ for $y \ge 1$ and $H_I= H_I(D_x, D_y) = H^S$ for $y \le -1$. Therefore $H_I$ models the interface $y\approx 0$ between two bulks governed by the constant coefficient Hamiltonians $H^N$ and $H^S$ and $\sigma_I(g)$ quantifies the signal transport in the $+x$ direction along the interface in the energy region $g$. We consider a flat interface for simplicity but curved interfaces may also be analyzed using this framework \cite{bal2023edge}.

The BEC under this formalism is a general principle stating that the number of edge states appearing in a particular (bulk insulating) energy region $g$, quantified by the edge current observable $2\pi \sigma_I(g)$, is related to the topological properties of the bulk Hamiltonians by the relation
\begin{equation}\label{eq:BEC}
 2\pi\sigma_I(g) = \BDI = \fC[P^S_\ell,P^N_\ell]
\end{equation}
where $\ell$ is a common spectral gap of the bulk Hamiltonians $H^h$ for $h\in\{N,S\}$ and $g$ is either a sub-interval or the entire interval of the common spectral gap $\ell$. In cases like the cold plasma model considered below where the BEC cannot be proven, $\sigma_I(g)$ can be evaluated from spectral calculations by the \textit{spectral flow}- the net number of spectral branches which span the spectral gap with positive slope \cite{bal2022topological}.

For discrete or periodic models, the BEC naturally arises from the presence of a compact Brillouin zone which mathematically links the spectral and topological properties of the system. In such systems the BEC has been proved in many general settings \cite{drouot2021microlocal, graf2013bulk,bernevig2013topological,prodan2016bulk, elbau2002equality, hatsugai1993chern}. Meanwhile, in continuum systems the Brillouin zone is essentially all of $\Rm^2$ so that the non-trivial compactification of the Brillouin zone significantly complicates any proof of the BEC. A class of continuous operators for which it is guaranteed to apply is that of {\em elliptic} operators \cite{bal2022topological,bal2024continuous,quinn2024approximations}. Elliptic operators are essentially characterized by singular values $|\lambda^h_j(\kp)|\to\infty$ as $|\kp|\to\infty$ for all branches $j$ and both $h\in\{N,S\}$; see above references. 

Unfortunately, the cold plasma model is not elliptic, and a number of studies have shown inconsistencies in the BEC at an interface between two differently-biased cold plasmas \cite{gangaraj2018coupled, hassani2020physical, han2022anomalous}. These inconsistencies consist of a number of edge states which either differs by one from the difference in bulk invariants $\mC[P^S_\ell] -\mC[P_\ell^N]$, or edge states which fail to span the entire insulating gap $\ell$. The above studies consider a sharp interface between two biased plasmas, or a sharp interface between a biased plasma and another topologically trivial material; i.e. $H_I(x, y, D_x, D_y) = H^S(D_x, D_y)$ for $y<0$ and $H_I(x, y, D_x, D_y) = H^N(D_x, D_y)$ for $y \ge 0$. Interestingly, we found recently that the BEC in photonic continua holds as long as \textit{bona fide} Chern numbers are defined via BDIs (i.e. condition \eqref{eq:gluing} holds) and the transition between $H^S$ and $H^N$ is continuous \cite{frazier2025topological}, suggesting that continuity of the boundary may  play an essential role in the BEC for continuum photonic systems. This was also the case in another continuum system, the shallow water equations which govern large-scale atmospheric oscillations. We briefly summarize the applicable results for this system below. 

\subsection{BEC for equatorial waves}\label{sec:3x3}
A system which displays similar anomalies in the BEC to continuum photonic systems is the linearized shallow water model, which predicts the topological protection of certain atmospheric oscillations concentrated at the equator. The Hamiltonian for the system, which acts on the state vector $(\eta, u, v)$ representing, respectively, atmospheric height, horizontal velocity, and vertical velocity, is (without regularizing terms) \cite{tauber2019bulk, bal2022topological}:
\[
H_I = 
\begin{pmatrix}
    0 & D_x & D_y\\
    D_x & 0 & -if(y)\\
    D_y & if(y) & 0
\end{pmatrix},
\]
where $f$ is a Coriolis force parameter which is negative in the southern hemisphere ($f(y) = f^S < 0$ for $y \le -1$) and positive in the northern hemisphere ($f(y) = f^N>0$ for $y \ge 1$). For a constant value of $f$, the spectra of the bulk Hamiltonians $\hat H^h(\bk)$, $h \in \{N, S\}$ are given by the bands: $E_0(\bk) = 0$ and $E_\pm(\bk) = \pm \sqrt{|\bk|^2 + f_h^2}$ \cite{delplace2017topological, tauber2019bulk}. Thus there are two global spectral gaps $(0, f_{min})$ and $(-f_{min}, 0)$ for $f_{min} = \min\{|f^S|, |f^N|\}$. BDIs for these spectral gaps, with respect to a transition from $H^S$ to $H^N$, can be defined as: $\fC_{\pm 1} = \pm 2$. We focus on the upper band gap for simplicity, although the results are equal but opposite for the lower band gap.

The BEC for this system was proved in any case where $f'(y)$ is bounded, i.e. $2\pi\sigma_I(g) = 2$ for any interval $g \subset (0, f_{min})$ given that $f'(y)$ is bounded for all $y$ \cite{bal2022topological}. However, when $f(y)$ has discontinuities at some discrete number of points $\{y_k\}_{k =1}^K \subset (-1, 1)$ the BEC no longer holds \cite{bal2024topological}. Denote the half-difference $f_{ko} = \frac 12(f(y_k^+)-f(y_k^-))$ and define the following sets of energies:
\[
\varepsilon_{L} = \{f_{ko} \;|\; 1 \le k \le K, \;f_{ko} >0\},
\quad
\varepsilon_{R} = \{-f_{ko}\;|\; 1 \le k \le K, \;f_{ko}<0\}.
\]
Further define $\mJ_L(E)$ as the number of indices $1\le k \le K$ such that $E_k\in \varepsilon_L$ and $E < E_k$ and $\mJ_R(E)$ as the number of indices $1\le k\le K$ such that $E_k \in \varepsilon_R$ and $E < E_k$. Suppose now that $g$ is a connected interval such that $g \subset (0, f_{min})\backslash (\varepsilon_L \cup \varepsilon_R)$ and which contains $E$. The main result of \cite{bal2024topological} is that:
\begin{equation}\label{eq_BEC}
2\pi\sigma_I(g) = 2-\mJ_L(E)+\mJ(E).
\end{equation}

Thus the edge current $\sigma_I(g)$ has contributions from both the invariant $\fC_{+1} = 2$ and from integer contributions of each discontinuity associated with $\mJ_L$ or $\mJ_R$. In particular, the energy range affected by each discontinuity is directly related to its half-jump value $f_{ko}$. This anomalous BEC was deduced by analyzing the spectral properties of $H_I$ explicitly, from which it was deduced that for each discontinuity $k$, an edge mode concentrated at $y_k$ exists whose asymptotic dispersion is $f_{ko}$ either as $k_x \to \infty$ or $-\infty$ depending on the sign of $f_{ko}$.

\subsection{BEC for continuum photonics}

In the cold plasma model for continuum photonics which we analyze below, the magnetic field parameter (cyclotron frequency) $\om(y)$ plays a similar role to the Coriolis force parameter $f(y)$ above, and we assume it may also contain a discrete number of discontinuities $\{y_k\}_{k = 1}^K$. Although the spectrum of $H_I$ is much more difficult to derive analytically in this case, we are able to deduce three (rather than one) edge states which concentrate around each discontinuity at high wavenumbers and have flat asymptotic dispersion. We denote their eigenvalues $\omega_e^{(j)}(k_x; y_k)$, $j \in \{1, 2, 3\}$. Each of these edge states can a priori contribute to the edge current of two separate spectral gaps, $(0, E_1)$, $(E_{uh}, E_2)$, via a flat asymptotic dispersion which depends on $\om(y_k^\pm)$ and the plasma frequency $\op$. We label the asymptotic values $\hat \omega_e^{(j)}(y_k) = \lim_{k_x \to \pm \infty}\omega_e^{(j)}(k_x; y_k)$ and define the sets of energies for $i \in \{1, 2, 3\}$:
\begin{align*}
\varepsilon_{iL} = \{\hat\omega_e^{(i)}(y_k)|\; 1 \le k \le K, \;\om(y_k^+)-\om(y_k^-) >0\}\nonumber\\
\varepsilon_{iR} = \{\hat\omega_e^{(i)}(y_k)|\; 1 \le k \le K, \;\om(y_k^+)-\om(y_k^-)<0\} \nonumber.
\end{align*}
Similarly $\mJ_{ij}(E)$ is the number of indices $1\le k\le K$ such that $E_k \in \varepsilon_{ij}$ and $E<E_k$ for $i \in \{1, 3\}$, $j \in \{L, R\}$, and $\mJ_{2j}$ is the number of indices $1\le k \le K$ such that $E_k \in \varepsilon_{2j}$ and $E> E_k$; $j \in \{L, R\}$. Our main result is that for $E \in (0, E_1)$:
\[
    2\pi \sigma_I(g) = \mJ_{1R}(E) -\mJ_{1L}(E)+\mJ_{2R}(E) -\mJ_{2L}(E)
\]
and for $E \in (E_{uh}, E_2)$:
\[
    2\pi \sigma_I(g) = 2+  \mJ_{3R}(E) -\mJ_{3L}(E)
\]
where $g \subset (0, E_1) \cup (E_{uh}, E_2) \backslash\left(\cup_{i = 1}^3 (\varepsilon_{iL} \cup \varepsilon_{iR})\right)$ and $g$ contains $E$. We will show in the next section that the well-defined BDIs for the bulk spectral gaps $(0, E_1)$ and $(E_{uh}, E_2)$ are, respectively, $(\fC_1, \fC_2) = (0, 2)$. Thus, each discontinuity within the interface alters the BEC by an integer amount depending on the energy region $g$. However, the spectral invariant $\sigma_I(g)$ still remains quantized as long as $g$ does not contain any of the discrete values $\cup_{i = 1}^3 (\varepsilon_{iL} \cup \varepsilon_{iR})$. Thus, although the traditional BEC does not hold, we still find that the remarkable spectral properties for which it is usually valued still apply. In the remainder of the paper we derive this result for the local Drude model for continuum photonics and subsequently extend our results to regularized models and the non-local hydrodynamic Drude model.

\section{BEC for local model}\label{sec:cold_plasma_model}
Our system of study will be the light-matter interacting cold plasma model, which models electromagnetic waves in an electron gas in the low-temperature limit. This model is valid in various energy regimes for electromagnetic wave propagation in various isotropic metals and semi-metals, with recent experimental successes for THz wave propagation in InSb \cite{wang2019photonic, liang2021tunable}. In this section non-local effects are ignored, producing the local Drude model, while in Section \ref{sec:hydro} we analyze the linearized hydrodynamic Drude model, which provides a minimal model of non-local effects. 

Coupling of the Lorentz force equation with Maxwell's equations in an electron gas biased by a spatially varying magnetic field with amplitude $B_0(x, y)$ in the $\hat z$ direction, and restricting wave propagation to the $xoy$ plane we obtain the following interface Hamiltonian for TM modes \cite{frazier2025topological, serra2025influence}:
\begin{equation}\label{eq:HTMTE}
     H_{I} = 
    \begin{pmatrix}
        0 & i\Omega(y) & i\omega_p & 0 & 0 \\
        -i\Omega(y) & 0 & 0 & i\omega_p & 0 \\
        -i\omega_p & 0 & 0 & 0 & -D_y \\
        0 & -i\omega_p & 0 & 0 & D_x \\
        0 & 0 & -D_y & D_x & 0 \\
    \end{pmatrix}, \qquad H_I\psi(y) = \omega\psi(y).
\end{equation}
$H_I$ acts on $\psi = (v_x, v_y, E_x, E_y, B_z)^T$ ($v$ for electron velocity, $E, B$ for electric and magnetic field) with the following definitions of the cyclotron frequency and plasma frequency:
\[\Omega(x,y) = \frac{q_e B_0(x,y)}{m_e},\quad \omega_p(x,y) =  \sqrt{\frac{n_eq_e^2 }{m_e\epsilon_0}}
\]
and the following normalizations:
\[ck\to k,\quad cB\to B,\ -v(t,x,y)\sqrt{\frac{m_e n_e}{\epsilon_0}} \to v(t,x,y).\]
Here $q_e, m_e$ are electron charge and mass and $n_e$ the average electron density, assumed to be constant. See \cite{han2022anomalous, frazier2025topological, parker2020topological, rajawat2025continuum, rajawat2026propagation} for models in which $n_e$ varies spatially. We assume that $\om(y) = \om_N$ for $y\ge 1$ and $\om(y) = \om_S$ for $y \le -1$ so that the bulk Hamiltonians $\hat H^h(\bk)$, $h \in \{N, S\}$ are given by $H_I$ with $\om(y) \to  \om_h$ and $(D_x, D_y) \to (k_x, k_y) = \bk$. It is easily verified that this system is equivalent to the a-priori more general photonic models in \cite{silveirinha2015chern,hanson2016notes, hassani2016effects,  silveirinha2016bulk}. Using the effective electron mass in lieu of $m_e$ for InSb and a typical electron density yields a plasma frequency of $\op \approx 2$ THz and shows that a relatively modest magnetic field can produce cyclotron frequencies in the THz range \cite{hassani2019truly}. Therefore we consider $\op$, $\om$, and $k$ normalized to THz and a typical value of $\op = 2$ in our analysis below. In most of the figures below we consider bulk modes with $|\om_{N/S}| = 0.5\op$, an achievable value for InSb \cite{liang2021tunable}.

The eigenvalues of $\hat H^h(\bk)$ are given by $\tilde\omega_{-1,-2}=-\tilde\omega_{1,2}$ and
\[\tilde\omega_{0,1,2}^2(k)= \left(0,\;\frac 12\left( 2\op^2+\om_h^2 + k^2 - \sqrt{(k^2-\om_h^2)^2 + 4\op^2\om_h^2}\right), \frac 12 \left(\op^2  + k^2 + \om_h^2 + \sqrt{(k^2-\om_h^2)^2 + 4\op^2\om_h^2} \right)\right)\] 
where $k = |\bk|$. Importantly, we observe that there are two global band gaps since $\tilde \omega_1(0) > 0$ and $\lim_{k \to \infty} \tilde\omega_1 = \sqrt{\op^2 + \om_h^2} < \tilde \omega_2( 0)$, with both bands monotonically increasing with $k$. Due to parity symmetry of the spectrum we analyze only the positive spectrum and band gaps. Integrals of Berry curvature over these bands yield \cite{hanson2016notes, serra2025influence}:
\begin{align}\label{eq:berry_int}
\mC_{\pm 1} = \mp\text{sgn}(\om_h)\left( 1 +\frac{\sigma_h}{\sqrt{1+\sigma_h^2}}\right)\\
\nonumber\\
\mC_{\pm2} = \pm \text{sgn}(\om_h)\nonumber
\end{align}
for $\sigma_h = |\om_h|/\op$ (note that the sign of $\om$ relative to $B_0$ changes the sign of the above invariants- here we use the conventions used in \cite{hanson2016notes}). Band gaps exist between bands 0 and 1 (labeled $\ell = 1$) and between bands 1 and 2 (labeled $\ell = 2$), and assuming that \eqref{eq:gluing} holds, we define BDIs for these two gaps $(\fC_1, \fC_2)$ via \eqref{eq:ChernBDI}.

We can see immediately that $\mC_{1}$ is not necessarily integer-valued, and can take on a continuum of values. This is a direct consequence of the fact that the limit of the projectors $\lim_{\bk \to \infty} \Pi_1(\bk)$ cannot be continuously defined so that $\mC_1$ is not a well-defined Chern number. A standard fix, introduced in \cite{silveirinha2015chern}, is to introduce a high-wavenumber regularization, for which in the parameters $\om$ and $\op$ have a non-trivial dispersion for high wavenumbers $k$; i.e. $\op, \om_h$ depend non-trivially on $k$ when $k$ becomes large. When $\om, \op$ are replaced by $\om(k), \op(k)$, $\sigma_h$ in \eqref{eq:berry_int} is replaced by $\lim_{k\to \infty} |\om_h(k)|/\op(k)$ so that, given this limit does not depend on the low-$k$ values of $\om_h$ and  $\op$, we may fix a stable value for $\mC_{\pm 1}$ and distinguish two distinct topological phases: $S$ ($\om < 0$) and $N$ ($\om > 0$). The physically-motivated SC regularization proposed in \cite{silveirinha2015chern} consists of the dispersion $\lim_{k\to \infty} \op(k) = 0$, while we found in \cite{frazier2025topological} that \eqref{eq:gluing} is satisfied for any regularization for which $\lim_{k\to\infty}\sigma(k) = 0$, therefore allowing us to define BDIs $(\fC_1, \fC_2) = (0, 2)$ which are \textit{bona fide} Chern numbers with respect to a transition between the two bulk phases $S$ ($\om < 0$) and $N$ ($\om>0$). Although the regularization $\op(k) \to 0$ produces stable invariants $(\fC_1, \fC_2) = (-2, 2)$ with respect to the same topological transition, the projectors associated to $\fC_1$ do not satisfy \eqref{eq:gluing} and therefore only $\fC_2$ is a well-defined Chern number for the SC regularization. Through numerical spectral calculations we showed the BDIs $(\fC_1, \fC_2) = (0, 2)$ correctly predict the number of edge states in each gap provided that $\om(y)$ \textit{continuously} transitions between the two phases, which is illustrated in Figure \ref{fig:cont_BEC}. See Section \ref{sec:reg_models} for further details on the two regularization methods and Appendix \ref{sec:numerics} for details regarding numerical methods.

\begin{figure}[t!]
    \centering
    \includegraphics[width = \textwidth]{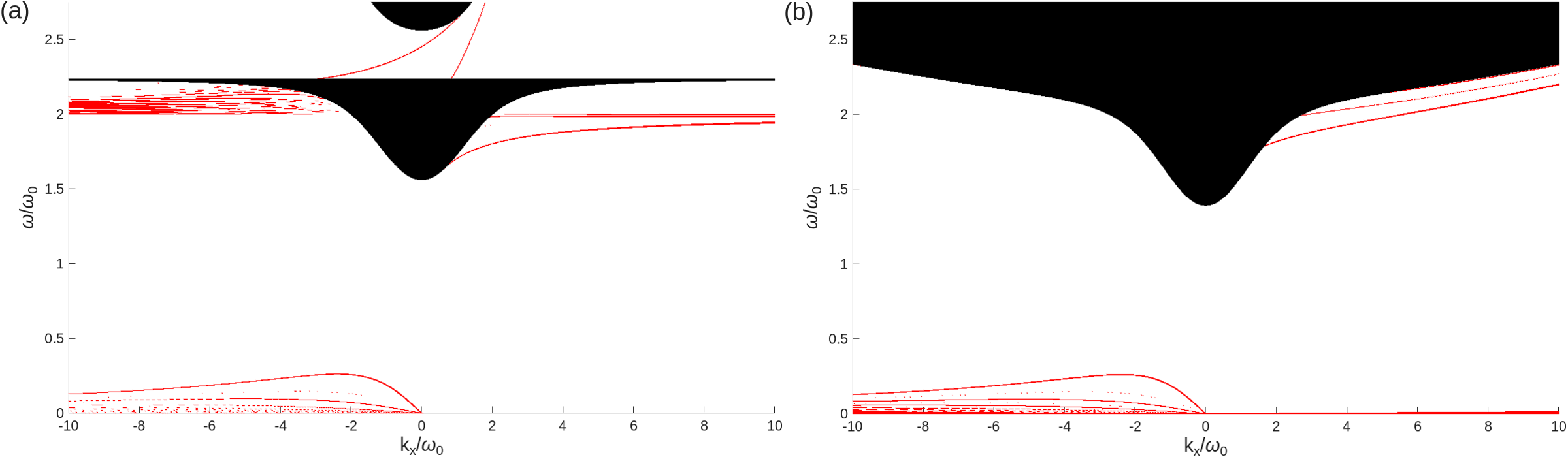}
\caption{Spectrum of $H_I(k_x)$ with $\om(y)$ which transitions smoothly between $\om_S = -1$ THz and $\om_N = 1$ THz with $\op = 2$ THz, $\omega_0 = 1$ THz. Black shading represents regions of bulk spectrum and red represents edge spectrum. (a) shows the local Drude model and (b) the hydrodynamic model discussed in Section \ref{sec:hydro}. In both cases the BEC holds using well-defined BDIs $(\fC_1, \fC_2) = (0, 2)$, although for the hydrodynamic model there is no longer an upper band gap.}
    \label{fig:cont_BEC}
\end{figure}

We now turn to the situation in which $\om(y)$ contains a single discontinuity at $y_0 \in (-1, 1)$. In the following we will additionally assume that $H_I$ is invariant in the $x-$direction so that we may apply a Fourier transform in $x$ only, giving $H_I = \mF^{-1}_{k_x \to D_x} H_I(k_x)\mF_{D_x \to k_x}$ for $D_x \to k_x$. This represents a flat boundary at $y \approx 0$ between the bulk systems $H^S$ and $H^N$ in the lower- and upper-half planes respectively. From the last paragraph we deduce global band gaps 1 and 2, respectively $(0, E_1)$ and $(E_{uh}, E_2)$ for $E_1 = \text{min}\{\tilde \omega_1^S(0), \tilde\omega_1^N(0)\}$, $E_{uh} = \sqrt{\op^2 + \text{max}\{\om^N, \om^S\}^2}$, and $E_2 = \text{min}\{\tilde \omega_2^S(0), \tilde\omega_2^N(0)\}$. We denote the following jump values for convenience:
\[
\om_+ = \om(y_0^+), \quad \om_- = \om(y_0^-), \quad \om_j = \frac{\om_+-\om_-}{2}.
\]

The spectrum of $H_I(k_x)$ is in general difficult to obtain analytically even for simple functions of $\om(y)$. However, at large values of $|k_x|$ we are able to deduce the presence of three eigenvalues of $H_I(k_x)$, $\omega_e^{(i)}(k_x)$, $i \in \{1, 2, 3\}$, which are associated to eigenmodes $\psi_{k_x}^{(i)}(y)$ of $H_I(k_x)$ that decay exponentially away from $y_0$. When $\om_j > 0$, $\omega_e^{(1, 3)}(k_x)$ exist for large negative values of $k_x$ and $\omega_e^{(2)}(k_x)$ exists for large positive values of $k_x$, and vice versa for $\om_j < 0$. While analytic expressions for $\omega_e^{(i)}(k_x)$ are not accessible, the limits $\hat \omega_e^{(2)} = \lim_{k_x\to \infty} \omega_e^{(2)}(k_x)$ and $\hat \omega_e^{(1, 3)} = \lim_{k_x \to -\infty}\omega_e^{(1, 3)}(k_x)$ correspond to the three positive roots (in ascending order) of the polynomial:
\begin{equation}\label{eq:asym_freq}
(\hat\omega_e^{(i)})^6 -(\hat\omega_e^{(i)})^4(2\op^2+\om_+^2+\om_-^2) + (\hat\omega_e^{(i)})^2 (\op^2+\om_+^2)(\op^2+\om_-^2) -\op^4\om_j^2 = 0, \quad i \in \{1, 2, 3\}.
\end{equation}

\begin{figure}[b!]
    \centering
    \includegraphics[width = \textwidth]{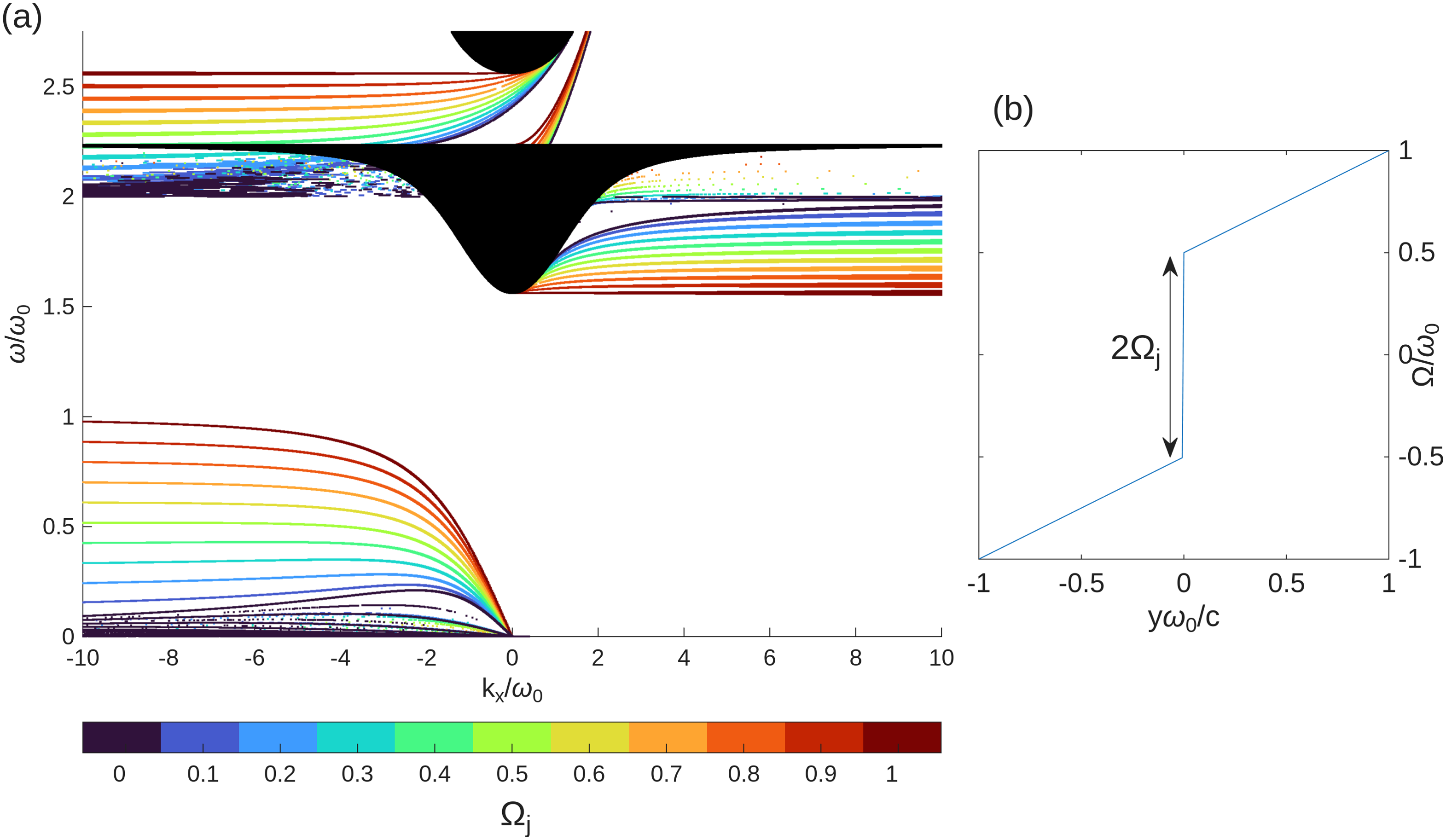}
    \caption{(a) Spectrum of $H_I(k_x)$ for which $\om(y)$ varies linearly between $\om(y\le -1) = \om^S = -1$ THz and $\om(y \ge 1) = \om^N = 1$ THz and has a single discontinuity at $y = 0$ for which $\om_j > 0$. Black regions correspond to the bulk spectrum and each color represents the edge spectrum for a different value of $\om_j$. (b) Plot of $\om(y)$ used in (a). $\op = 2$ THz, $\omega_0 = 1$ THz.}
    \label{fig:BEC_jumps}
\end{figure}

\begin{figure}[ht!]
    \centering
    \includegraphics[width = \textwidth]{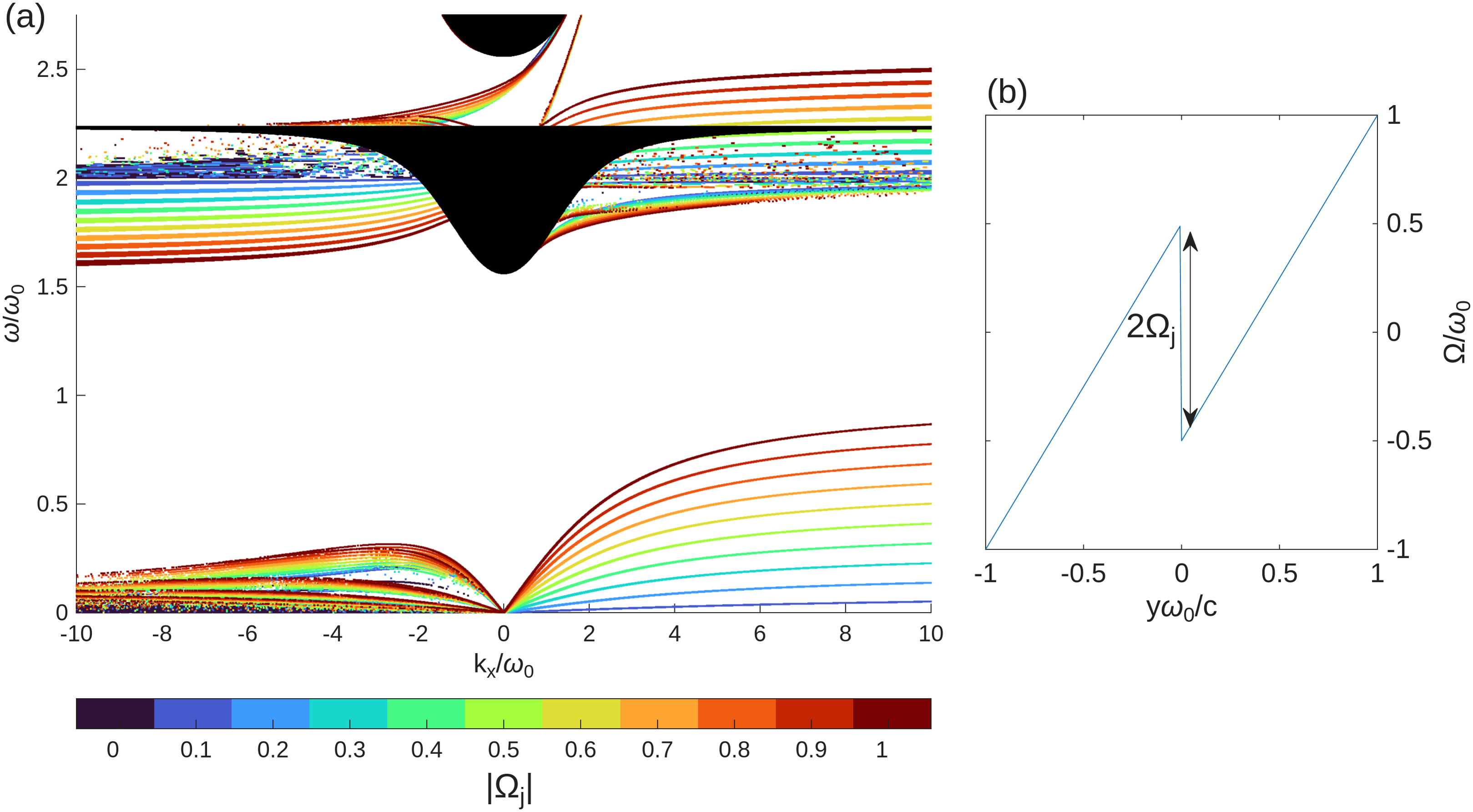}
    \caption{(a) Spectrum of $H_I(k_x)$ for which $\om(y)$ varies linearly between $\om(y\le -1) = \om^S = -1$ and $\om(y \ge 1) = \om^N = 1$ and has a single discontinuity at $y = 0$ for which $\om_j < 0$. Black regions correspond to bulk spectrum and each color represents the edge spectrum for a given value of $|\om_j|$. (b) Plot of $\om(y)$. $\op = 2$ THz, $\omega_0 = 1$ THz.}
    \label{fig:BEC_jumps_neg}
\end{figure}

Note in particular that $\hat\omega_e^{(1)}, \hat\omega_e^{(2)}, \hat\omega_e^{(3)}$ depend only on the values of $\om_-,\om_+$ and $ \op$. A fourth edge mode, $\omega_e^{(4)}(k_x)$, exists whenever sgn$(\om_N) = -\text{sgn}(\om_S)$ and converges to $+\infty$ as $k_x \to \infty$ whenever $\om_j > 0$ and converges to $+\infty$ as $k_x \to -\infty$ when $\om_j < 0$. $\omega_e^{(4)}(0) \in (E_1, E_{uh})$ and thus this mode contributes the same spectral flow regardless of the regularity of $\om(y)$. See Appendix \ref{sec:5x5_analysis} for a derivation of the above results.

While we have derived the existence of at least four spectral branches which may contribute to asymmetric edge transport, we turn to numerical spectral calculations to resolve the full spectrum of $H_I$ for general profiles of $\om(y)$. We find two regions of edge spectrum which are not described by the above analysis and occur only when $\om'(y)\ne 0$ for some interval in $-1 < y < 1$. First, we observe a continuum of edge states within the energy region $ (\op, E_{uh})$. In \cite{rajawat2025continuum} these modes were identified as ``local upper hybrid resonances" whose energies correspond to $\sqrt{\op^2 + \om^2(y)}$ for $y \in (-1, 1)$ and whose eigenfunctions are concentrated around the value of $y$ corresponding to the associated local upper hybrid frequency. Numerical spectral calculations from \cite{parker2020topological} also support this explanation. In our model of TM wave propagation these modes are confined outside the two band gaps and therefore do not contribute to the spectral flow of either, except when applying the SC regularization. Another set of bands of edge spectrum originate at the point $(k_x, \omega) = (0, 0)$ and also converge to 0 as $k_x \to \pm \infty$ depending on if $\om_j > 0$, or $\om_j < 0$. Since they only appear when $\om'(y) \ne0$ and are localized near $\omega = 0$, we expect that these branches arise from higher order ($\om'(y)$) interactions between the trivial band at $\tilde \omega_0 = 0$ and the longitudinal bulk band $\tilde \omega_1(k_x)$, in analogy with Rossby waves which arise in the shallow water system. As shown in Figures \ref{fig:BEC_jumps}-\ref{fig:BEC_2_jumps}, both sets of edge states, which only appear when $\om'(y)$ is non-zero, do not contribute to the spectral flow of either band gap except when applying the SC regularization, which we will discuss in the next section. 

Therefore, for each jump in $\om(y)$ three edge states, $\omega_e^{(1)}, \omega_e^{(2)}, \omega_e^{(3)}$, may alter the BEC for the two spectral gaps depending on their asymptotic dispersion as $k_x \to \pm \infty$. We now assume that a finite number of discontinuities in $\om(y)$ exist at points $\{y_k\}_{k = 1}^K \subset (-1, 1)$ and denote $\om_\pm^{(k)} = \om(y_k^\pm)$. Similarly to the shallow water model we define the following sets of frequencies for $i \in \{1, 2, 3\}$:
\begin{align}\label{eq:jump_energy}
\varepsilon_{iL} = \{\hat\omega_e^{(i)}(\om_+^{(k)}, \om_-^{(k)})|\; 1 \le k \le K, \;\om_j^{(k)} >0\}\\
\varepsilon_{iR} = \{\hat\omega_e^{(i)}(\om_+^{(k)}, \om_-^{(k)})|\; 1 \le k \le K, \;\om_j^{(k)}<0\}.\nonumber
\end{align}
Define $\mJ_{ij}(E)$ as the number of indices $1\le k\le K$ such that $E_k \in \varepsilon_{ij}$ and $E<E_k$ for $i \in \{1, 3\}$, $j \in \{L, R\}$, and $\mJ_{2j}$ as the number of indices $1\le k \le K$ such that $E_k \in \varepsilon_{2j}$ and $E> E_k$; $j \in \{L, R\}$. Suppose now that $g$ is any connected interval such that $g \in (0, E_1) \cup (E_{uh}, E_2) \backslash (\cup_{i = 1}^3 \varepsilon_{iL} \cup \varepsilon_{iR})$ and which contains $E$. We find that for $E \in (0, E_1)$:
\begin{equation}\label{eq:BEC_1}
    2\pi \sigma_I(g) = \mJ_{1R}(E) -\mJ_{1L}(E)+\mJ_{2R}(E) -\mJ_{2L}(E)
\end{equation}
and for $E \in (E_{uh}, E_2)$:
\begin{equation}\label{eq:BEC_2}
    2\pi \sigma_I(g) = 2+  \mJ_{3R}(E) -\mJ_{3L}(E).
\end{equation}
Note that these equations contain contributions to $\sigma_I(g)$ from the well-defined BDIs $(\fC_1, \fC_2) = (0, 2)$ and from the flat asymptotic dispersion of the bands $\omega_e^{(i)}(k_x)$, $i \in \{1, 2, 3\}$. This result is proved in limited cases in Appendix \ref{sec:5x5_analysis} and we turn to numerical spectral calculations to verify it in general. In most cases, we will consider the symmetric case $\om_- = -\om_+$ with $\om_j^2 < \text{max}\{\om_N^2, \om_S^2\}$ so that (from \eqref{eq:asym_freq}) $E_1<\hat\omega_e^{(2)}< E_{uh}$. In this case for $E \in (0, E_1)$ we have that $\mJ_{2L}(E), \mJ_{2R}(E)= 0$ and therefore:

\begin{equation}\label{eq:BEC_1_simp}
2\pi \sigma_I(g) = \mJ_{1R}(E)-\mJ_{1L}(E), \qquad 0< E< E_1.
\end{equation}

\begin{figure}[t!]
    \centering
    \includegraphics[width = \textwidth]{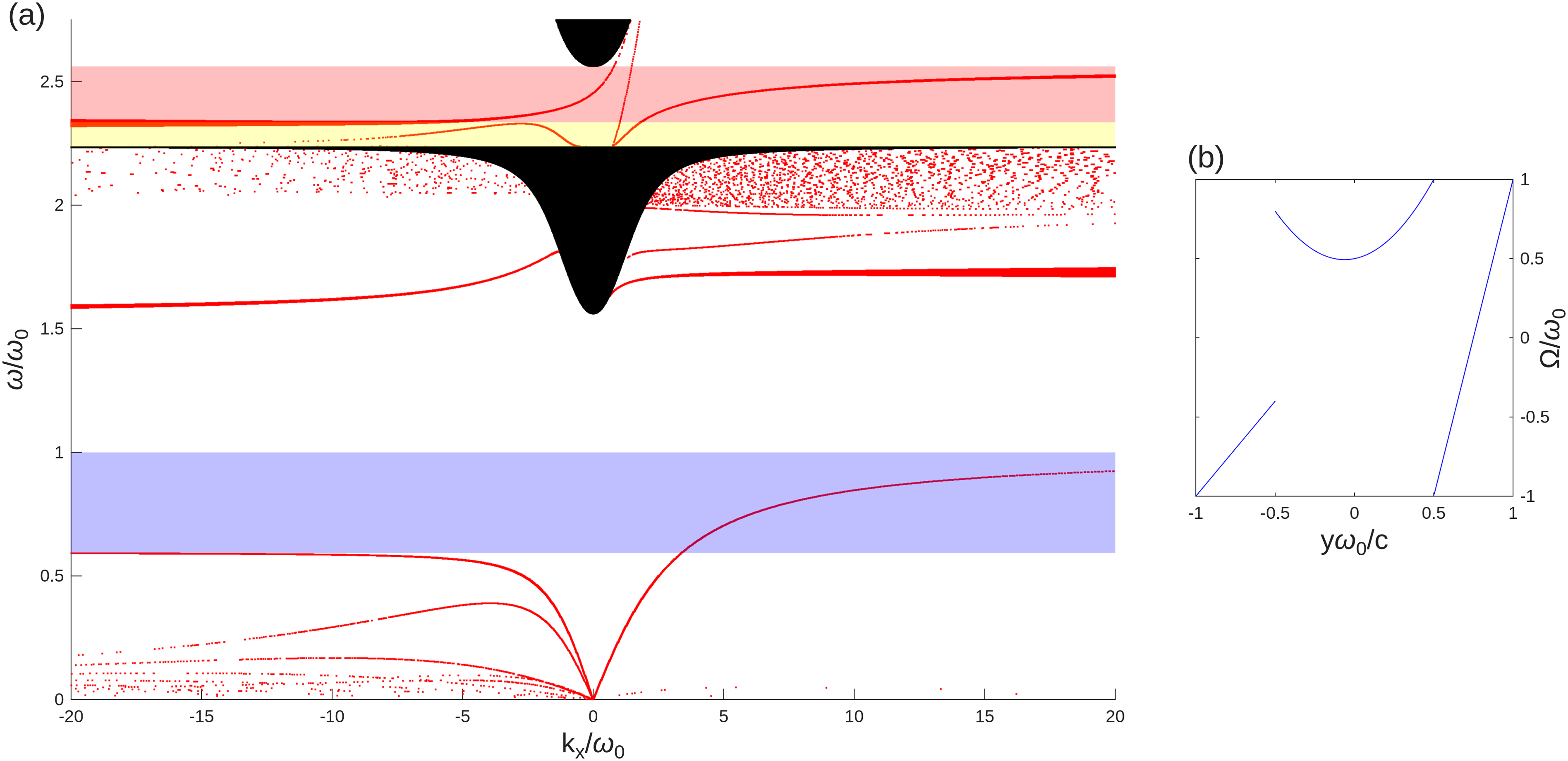}
    \caption{(a) Spectrum of $H_I(k_x)$ with $\om(y)$ as shown in (b). Areas of black represent bulk spectrum and areas of red represent edge spectrum. (b) $\om(y)$ which contains two discontinuities, one with $\om_j > 0$ and one with $\om_j < 0$, and one of which is not symmetric about $y = 0$. The spectral flow is now modified by both $\mJ_{iR}$ and $\mJ_{iL}$, $i\in\{1, 3\}$. For the first jump $\om_-^{(1)} = -0.4$, $\om_+^{(1)} = 0.8$ which yields $\hat\omega_e^{(1)} = 0.5941$ and $\hat\omega_e^{(3)} = 2.3356$. For the second jump $\om_-^{(2)} = -1$ and $\om_+^{(2)} = 1$ yielding $\hat\omega_e^{(1)} = 1$ and $\hat\omega_e^{(3)} = 2.5616 = E_2$. Therefore the blue shaded region $(0.5941, 1)$ now has a spectral flow of 1, the yellow region $(E_{uh}, 2.3356)$ a spectral flow of +2, and the red region $(2.3356, E_2)$ a spectral flow of +3. $\op = 2$ THz, $\omega_0 = 1$ THz.}
    \label{fig:BEC_2_jumps}
\end{figure}

Equations \eqref{eq:BEC_2} and \eqref{eq:BEC_1_simp} are verified numerically in Figure \ref{fig:BEC_jumps} and Figure \ref{fig:BEC_jumps_neg}, in which we consider $\om(y)$ which is linear between $\om^S = -1$, $\om^N = 1$ and has one discontinuity at $y = 0$ varying in magnitude from $|\om_j| = 0$ (no discontinuity) to $|\om_j| = 1$ (piecewise constant case), and in particular $\om_- = -\om_+$. Figure \ref{fig:BEC_jumps} shows the case where $\om_j > 0$ and Figure \ref{fig:BEC_jumps_neg} where $\om_j < 0$. Applying equations \eqref{eq:BEC_2} and \eqref{eq:BEC_1_simp} for $\om_j > 0$ gives $2\pi\sigma_I = -1$ in the region $0 <E < \om_j$, $2\pi \sigma_I = 0$ for $\om_j < E < E_1$, $2\pi\sigma_I = +1$ for $E_{uh} < E< \hat \omega_e^{(3)}(\om_j)$, and $2\pi \sigma_I = +2$ for $\hat \omega_e^{(3)}(\om_j) < E<  E_2$, where we have used the fact that $\hat \omega_e^{(1)} = \om_j$ when $\om_- = -\om_+$. Similarly for $\om_j < 0$ we get $2\pi\sigma_I = +1, 0, +3, +2$ respectively for the above mentioned energy ranges. Comparison with spectral flows in Figures \ref{fig:BEC_jumps} and \ref{fig:BEC_jumps_neg} shows that these predictions are correct. To show that our results hold in a more general case, Figure \ref{fig:BEC_2_jumps} shows an $\om(y)$ which contains two discontinuities and is non-linear.

Note that the invariants $(\fC_1, \fC_2)$ are defined by applying a high-wavenumber regularization to \eqref{eq:HTMTE}, however we have analyzed the BEC through the spectral flow of the unregularized system. We justify this by taking the view that the primary motivation of high-wavenumber regularization is to correct ill-defined topology at $\bk \to \infty$ and not to model any physical phenomenon so that the (arbitrarily) high wavenumber effects of regularization on the spectrum of $H_I$ can be disregarded (the situation is somewhat more nuanced for the SC regularization, which has a physical interpretation- see Section \ref{sec:reg_models}). In particular we note that even the non-local hydrodynamic model considered in Section \ref{sec:hydro} is derived from a semi-classical model and therefore is invalid for wavelengths comparable to or smaller than the Fermi wavelength. In any case we present a mathematically self-consistent analysis of regularized models in the next section.

\section{Regularized Models}\label{sec:reg_models}

In this section we analyze the BEC in a self-consistent way by comparing topological invariants to the spectral properties of the regularized Hamiltonians themselves. Critically, all the results of the previous section hold with the exception that $\om_\pm, \op$ in \eqref{eq:asym_freq} are replaced by their regularized values $\lim_{k\to \infty}\om_\pm(k), \op(k)$ according to equations \eqref{eq:reg_BDI} and \eqref{eq:reg_SC} below. 

The two regularizations we analyze are the BDI regularization, introduced more recently in \cite{frazier2025topological}, and the ubiquitous SC regularization introduced in \cite{silveirinha2015chern} and used extensively in the literature. The primary motivation for the BDI regularization is to ensure \eqref{eq:gluing} is met for the local Drude model so that \textit{bona fide} BDIs can be defined. We found in \cite{frazier2025topological} that \eqref{eq:gluing} was met for any regularization in which $\lim_{k \to \infty}\sigma(k) = 0$, and for concreteness we choose the regularization termed the BDI regularization in the rest of this paper as:
\begin{equation}\label{eq:reg_BDI}
\om(k, y) = \om(y)\left(1 + \left(\frac {k}{k_c}\right)^2\right)^{-1/2},
\end{equation}
where $k_c$ is an arbitrary cutoff wavenumber. In addition, we provided some heuristic justification for \eqref{eq:reg_BDI} in \cite{frazier2025topological} by considering some of the high-wavenumber effects of non-local considerations in the local Drude model. Note that since the topological invariants $\fC_j$ depend only on the properties of the bulk Hamiltonians, it is mathematically feasible to have a regularization which only applies for $|y|>1$. However, we apply the regularization at all $y$ values, which eliminates the somewhat pathological case that $\om(k, y)\ne 0$, $y\in (-1, 1)$ as $k_x\to \infty$ even as $\om^N(k), \om^S(k) \to 0$. As described in Section \ref{sec:cold_plasma_model} the (well-defined) BDIs produced by the BDI regularization are $(\fC_1, \fC_2) = (0, 2)$. 

We also analyze the ubiquitous Spatial Cutoff (SC) regularization:
\begin{equation}\label{eq:reg_SC}
\op(k) = \op\left(1 + \left(\frac {k}{k_c}\right)^2\right)^{-1}.
\end{equation}
This regularization is motivated by the physical principle that as the wavelength of EM excitations becomes small compared to the average distance between electrons, the macroscopic material response should vanish. Thus as $k\to \infty$ the SC regularization ensures that $\op\to 0$ so that EM wave propagation ($(E_x, E_y, B_z)$ for TM wave propagation) is de-coupled from material electron excitations ($(v_x, v_y)$ for TM wave propagation) at high wavenumbers, as seen from \eqref{eq:HTMTE}. In addition, this regularization applied to the local Drude model can be used to model the interface between a gyrotropic material and an unbiased plasma with an air gap of size $\sim 1/k_c$ in between \cite{silveirinha2016bulk}.

The invariants produced by the SC regularization are $(\fC_1, \fC_2) = (-2, 2)$. Although we found in \cite{frazier2025topological} that the SC regularization does not produce a well-defined Chern number for $\fC_1$, the invariants produced are both integers so that an evaluation of the BEC is still possible for both band gaps. Note that there is a considerable body of work on evaluating the BEC for the SC regularization for various phase transitions \cite{gangaraj2018coupled, silveirinha2015chern,   hassani2019truly, buddhiraju2020absence, han2022anomalous, silveirinha2016bulk}, but all previous studies have only considered the case where $\om, \op$ are piecewise constant with the exception of \cite{rajawat2025continuum, rajawat2026propagation}, which consider the case that $n_e$ varies linearly from a constant bulk value to zero across a finite interface.

In practice $k_c$ should be considered arbitrarily large so that the regularization produces well-defined topological invariants without altering the behavior of the system at physically relevant wavenumbers. However, we consider relatively small values of $k_c$ in order to demonstrate the spectral properties of the regularized model while avoiding high-wavenumber regions in which our numerical method may be ill-conditioned (see Appendix \ref{sec:numerics} for details on numerical methods). See e.g. \cite{silveirinha2016bulk, han2022anomalous} for analysis of edge states as $k_c \to \infty$. 

\subsection{BDI regularization}\label{sec:bdi_reg}

\begin{figure}[b!]
    \centering
    \includegraphics[width = \textwidth]{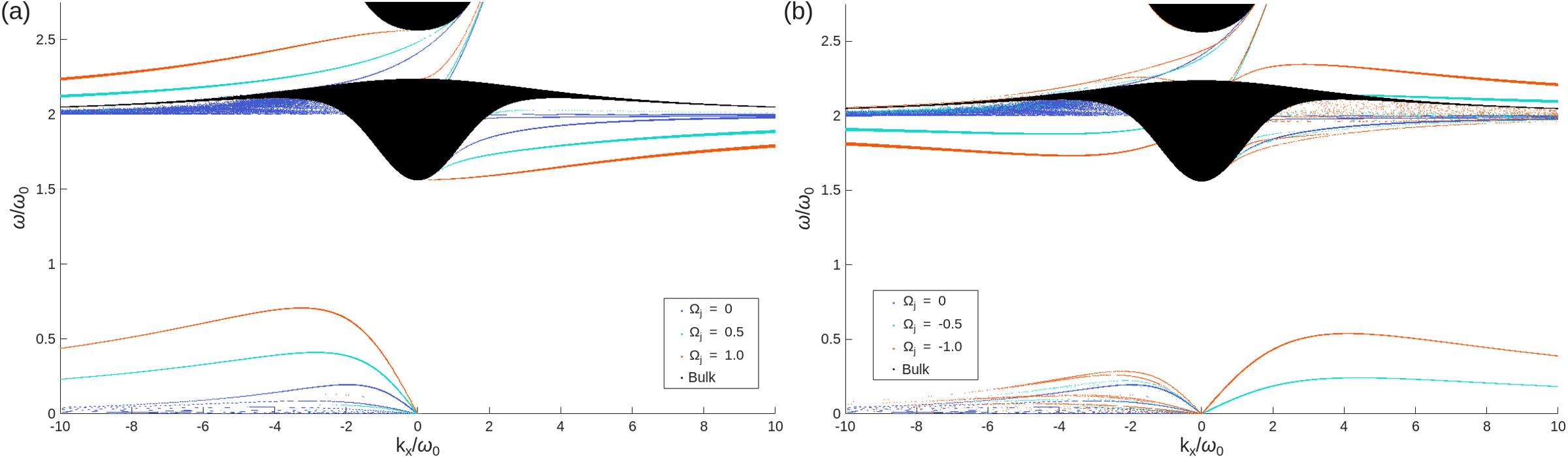}
\caption{Numerically calculated spectra for local model with BDI regularization \eqref{eq:reg_BDI} applied with $k_c = 5$ THz. $\om(y)$ is a linear transition between $\om^S = -1$ and $\om^N = 1$ with a jump at $y = 0$ as in Figures \ref{fig:BEC_jumps} and \ref{fig:BEC_jumps_neg}. (a) models positive jumps $\om_j > 0$ and (b) negative jumps $\om_j < 0$. $\op = 2$ THz, $\omega_0 = 1$ THz.}
\label{fig:bdi_reg}
\end{figure}
Applying \eqref{eq:reg_BDI} to the asymptotic dispersion relations \eqref{eq:asym_freq} we find that $\lim_{k\to \infty}\om_\pm(k) = 0$ so that for the BDI regularization: 
\[
\hat\omega_e^{(1)} =0, \quad
 \hat\omega_e^{(2)} = \op, \quad
\hat\omega_e^{(3)} =\op.
\]
Therefore we find that for the BDI regularization when $E \in (0, E_1)$, $\mJ_{1j}(E), \mJ_{2j}(E) = 0$ and when $E \in (E_{uh}, E_2)$, $\mJ_{3j}(E) = 0$, $j \in \{L, R\}$. This suggests that the BEC is \textit{always} satisfied for regularization \eqref{eq:reg_BDI}, which produces BDIs of $(\fC_1, \fC_2) = (0, 2)$. This is confirmed by numerical evidence in Figure \ref{fig:bdi_reg}, where we observe that $\omega_e^{(1)}(k_x)$ and $\omega_e^{(3)}(k_x)$, which previously converged to asymptotes $\om_j$ and $\hat \omega_e^{(3)}$ respectively, are now deformed to 0 and $\op$ by the regularization so that they no longer contribute to the spectral flow. 

This is additional evidence that regularization \eqref{eq:reg_BDI} provides a more robust quantization of edge modes than the SC regularization, which we consider below. While in \cite{frazier2025topological} we found that the (well-defined) BDIs $(\fC_1, \fC_2) = (0, 2)$ applied to the unregularized local model still predict the correct number of edge modes when $\om(y)$ is continuous, we find here that when we apply \eqref{eq:reg_BDI} in a self-consistent way to the regularized interface Hamiltonian itself, the aforementioned BDIs predict the correct number of edge states even in the presence of discontinuities in $\om(y)$. 

\subsection{SC regularization}\label{sec:sc_reg}

The SC regularization has been used extensively to study continuum photonic systems by providing a physically motivated means to restore integer values to integrals of Berry curvature, allowing the BEC to be evaluated. Although this regularization does not produce a well-defined Chern number for the lower band gap it has the advantage of a physical interpretation \cite{silveirinha2016bulk} and we wish to extend previous results for piecewise constant parameters $\om, \op$ to more general profiles of $\om(y)$. 

First, the SC regularization produces a separate set of invariants for the two band gaps present in TM wave propagation. Assuming \eqref{eq:reg_SC} we find that $\lim_{k\to \infty} \sigma(k) = \infty$ which produces integrals of Berry curvature $\mC_{\pm 1} = \pm2\text{sgn}(\om)$, $\mC_{\pm2}=\mp\text{sgn}(\om)$, for which we obtain $(\fC_1,\fC_2) = (-2, 2)$. The projectors associated to $\fC_2$ satisfy \eqref{eq:gluing}, however those associated with $\fC_1$ do not (see \cite{frazier2025topological} Appendix F for details). If we apply \eqref{eq:reg_SC} to \eqref{eq:asym_freq} we now obtain $\lim_{k\to \infty}\op(k)  =0$ so that we obtain for the symmetric case $\om_N = -\om_S$:
\[
\hat \omega_e^{(1)} = \om_j, \quad\hat\omega_e^{(2)} = 0, \quad \hat \omega_e^{(3)} = \om_j.
\]
Here we see that the asymptotic dispersion of $\omega_e^{(1)}$ remains unchanged, but the asymptotic dispersion of $\omega_e^{(2)}, \omega_e^{(3)}$ are altered with $\omega_e^{(2)}$ no longer depending on $\om_j$. Thus, assuming as above that $|\om_j| < \text{max}\{|\om_N|, |\om_S|\}$ we find that for $E \in (E_{uh}, E_2)$, $\mJ_{3j}(E) = 0$ and the BEC holds in the upper band gap. However, for the lower band gap we find that each discontinuity has a contribution of 1 to either $\mJ_{2L}(E)$ or $\mJ_{2R}(E)$ depending on the sign of $\om_j$ for any $E \in (0, E_1)$, while $\mJ_{1j}(E)$ remains the same as in the unregularized case. Thus in the region $E \in (0, \om_j)$ we expect that the BEC holds when using the SC invariant $\fC_1 = -2$ only when one jump is present. This is indeed the case when $\om(y) = \text{sgn}(y)\om_N$ so that $\om_j = \om_N$, which is the case treated in e.g. \cite{gangaraj2018coupled}. Since $E_1 < |\om_N| <E_{uh}$ we find that the BEC in fact holds for this special case, shown in Figure \ref{fig:sc_reg_pos_jump}b.

\begin{figure}[b!]
    \centering
    \includegraphics[width = \textwidth]{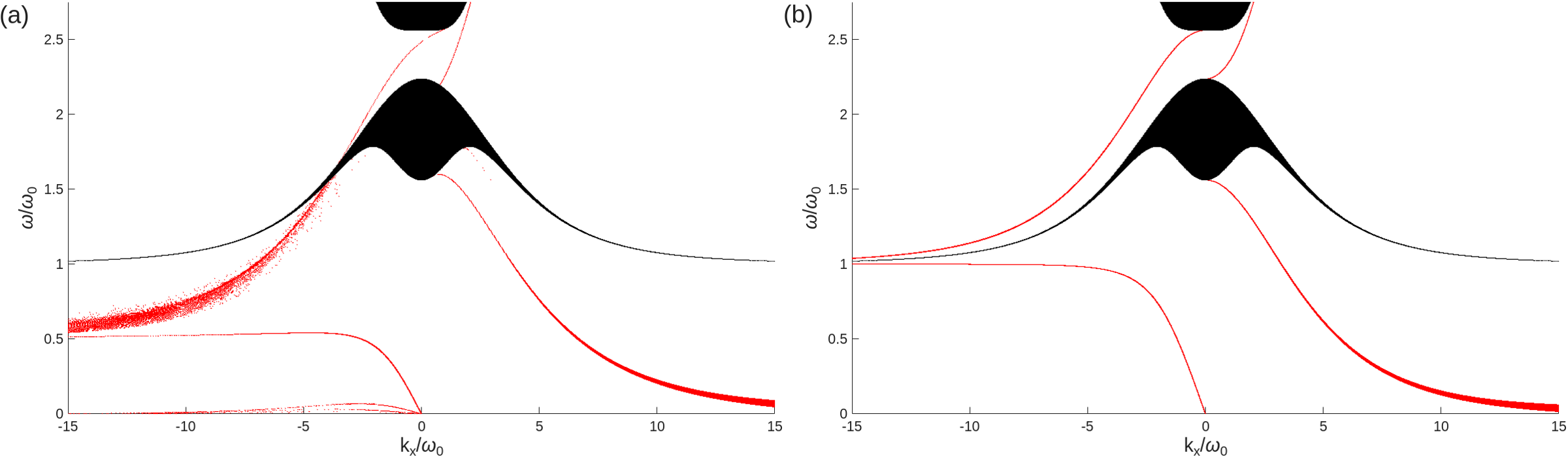}
\caption{Spectra calculated using SC regularization \eqref{eq:reg_SC} for a linear transition between $\om^S = -1$ THz, $\om^N = 1$ THz with a single positive jump at $y = 0$ and $k_c = 5$ THz. Areas of black indicate bulk spectrum and areas of red indicate edge spectrum. (a) shows $\om_j = 0.5$ THz and (b) $\om_j = 1.0$ THz. We can see the BEC is satisfied in the upper band gap and the region $(0, \om_j)$. $\op = 2$ THz, $\omega_0 = 1$ THz.}
\label{fig:sc_reg_pos_jump}
\end{figure}

\begin{figure}[t!]
    \centering
    \includegraphics[width = \textwidth]{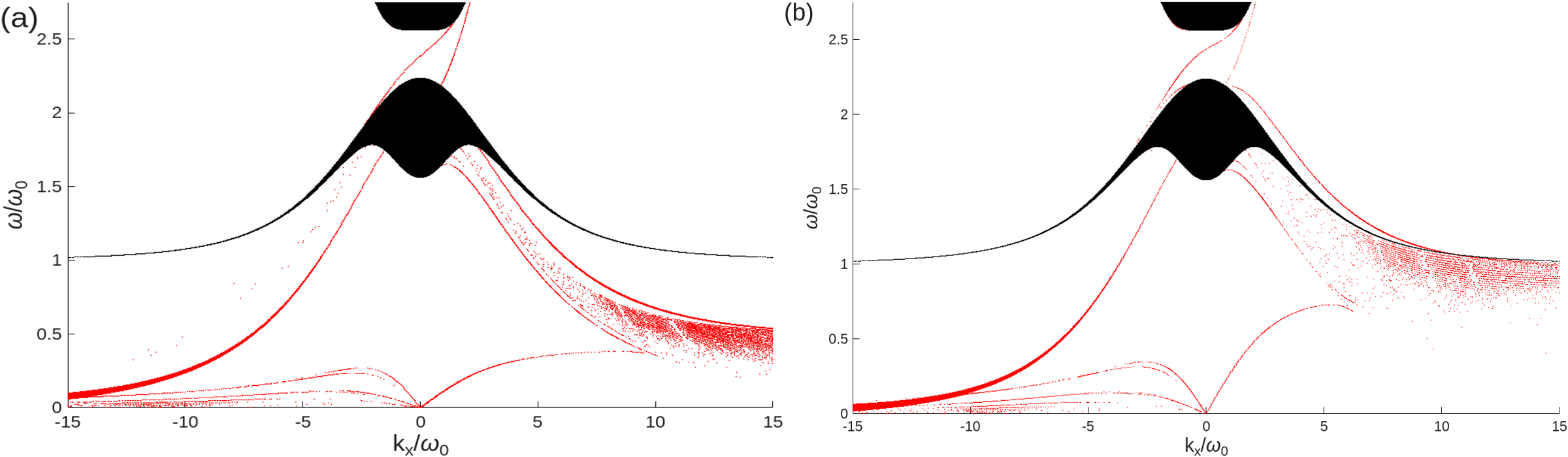}
\caption{Spectra calculated using SC regularization \eqref{eq:reg_SC} for a linear transition between $\om^S = -1$ THz, $\om^N = 1$ THz with a single positive jump at $y = 0$ and $k_c = 5$ THz. Areas of black indicate bulk spectrum and areas of red indicate edge spectrum. (a) shows $\om_j = -0.5$ THz and (b) $\om_j = -1.0$ THz. The BEC holds in the upper band gap while the lower band gap is effectively closed by a continuum of edge modes. $\op = 2$ THz, $\omega_0 = 1$ THz.}
\label{fig:sc_reg_neg_jump}
\end{figure}

\begin{figure}[h!]
    \centering
    \includegraphics[width = 0.6\textwidth]{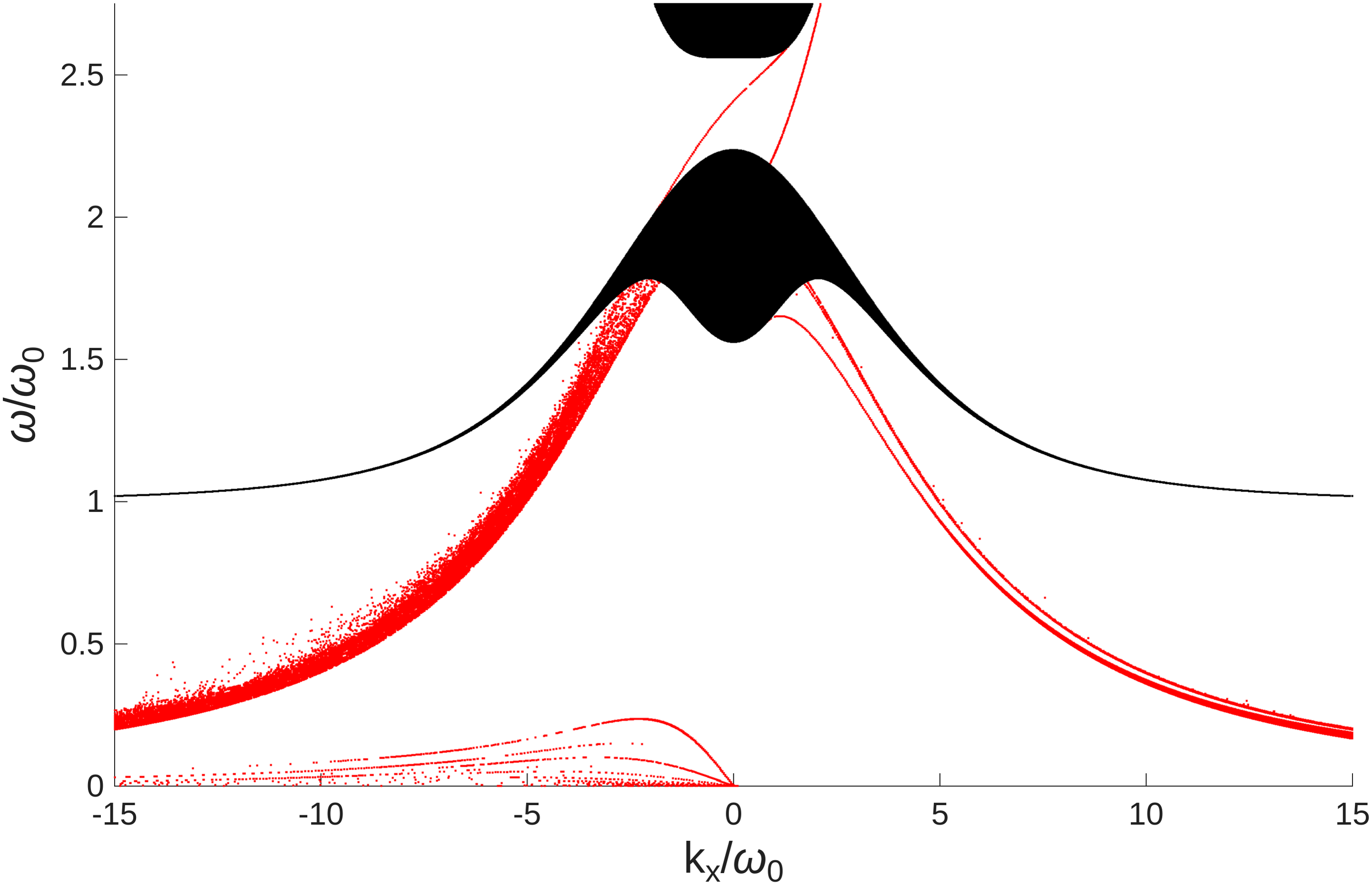}
\caption{SC regularization spectrum for a continuous linear transition between $\om^S = -1$ THz and $\om^N = 1$ THz with $k_c = 5$ THz. The lower band gap is effectively closed by the continuum of modes deformed to zero by the regularization as $k_x \to \infty$. $\op = 2$ THz, $\omega_0 = 1$ THz.}
\label{fig:sc_reg_linear}
\end{figure}

However, the picture becomes much less clear when considering all the edge modes in the general case $\om'(y) \ne 0$. The continuum of modes appearing in the region $(\op, E_{uh})$ for the unregularized system is now deformed into the lower band gap by the SC regularization. When considering a single positive jump, $\om_j> 0$ we can see from Figure \ref{fig:sc_reg_pos_jump} that this continuum of modes is contained above $\om_j$ so that for a single jump the BEC is satisfied in the regions $(0, \om_j)$ and $(E_{uh}, E_2)$ with the SC invariants $(\fC_1, \fC_2) = (-2, 2)$. In particular the piecewise case in Figure \ref{fig:sc_reg_pos_jump}b agrees with previous results \cite{silveirinha2016bulk, serra2025influence} and in the continuous case (Figure \ref{fig:sc_reg_linear}) the lower band gap effectively ceases to exist under the SC regularization. Interestingly, a single negative jump $\om_j < 0$ produces similar results to the continuous case, where the continuum of edge modes is deformed into the entirety of the lower band gap, making it impossible to evaluate the BEC there.  

Therefore we again observe in the SC regularized system that the BEC holds unambiguously (for any interface regularity) in the upper band gap for the well-defined BDI $\fC_2 = 2$. For the lower band gap the invariant $\fC_1 = -2$ is not a well-defined BDI. We see in the lower gap that the BEC is altered by discontinuities in the interface and in addition the regularization effectively alters the band gap such that it is impossible to evaluate the spectral flow when the boundary is continuous or $\om_j < 0$.

\section{Hydrodynamic model}\label{sec:hydro}

For a minimal model of non-local effects in the cold plasma model we consider the overall electron density as the average density $n_e$ (which does not depend on $t$) plus some small variation caused by the plasma oscillation, $\rho(t,x, y)$. Supplementing Maxwell's equations now with a hydrodynamic continuity equation \cite{pakniyat2022chern}, the linearized model of TM wave propagation in the hydrodynamic Drude model is represented by the pseudo-Hamiltonian \cite{serra2025influence}:
\begin{equation}\label{eq:H_TM_hydro}
\hat H^\beta(\bk) = \begin{pmatrix}
    0 & i\om & -i\op & 0 & 0 & \beta k_x\\
    -i\om & 0 & 0 & -i\op & 0 & \beta k_y\\
    i\op & 0 & 0 & 0 & -k_y & 0\\
    0 & i\op & 0 & 0 & k_x & 0\\
    0 & 0 & -k_y & k_x & 0 & 0\\
    \beta k_x & \beta k_y & 0 & 0 & 0 & 0
\end{pmatrix} 
\end{equation}
which acts on the components $(v_x, v_y, E_x, E_y, B_z, \rho)^T$. $\beta$ is a non-local parameter normalized to units of $c$ with typical values of $\beta \lesssim 10^{-2}$ \cite{hassani2019truly}. In reality $\beta$ is complex, with an imaginary part which models dissipative effects, however we consider a lossless setting here $\beta \in \Rm$ to ensure $\hat H^\beta(\bk)$ and $H^\beta_I$ are self-adjoint, a necessary condition for our analysis in Section \ref{sec:background} and the supporting theory.

$\hat H^\beta(\bk)$ has parity symmetry (see below) and therefore has three non-negative spectral bands and three non-positive bands. There are in fact two trivial bands $\tilde \omega = 0$ and four bands which are symmetric about zero \cite{serra2025influence}:
\begin{equation}\label{eq:bulk_hydro_1}
\tilde \omega_{\pm1}^2(k) = \frac 12 \left(2\op^2 + (1+\beta^2)k^2 +\om^2 -\sqrt{((1-\beta^2)k^2-\om^2)^2+4\op^2\om^2}\right)
\end{equation}
\begin{equation}\label{eq:bulk_hydro_2}
\tilde \omega_{\pm2}^2(k) = \frac 12 \left(2\op^2 + (1+\beta^2)k^2 + \om^2 +\sqrt{((1-\beta^2)k^2-\om^2)^2+4\op^2\om^2}\right).
\end{equation}

A number of symmetries, denoted below, are useful in our analysis of Chern numbers for the system. First, denoting $H(\theta) = \hat H^\beta(\bk)$ with $\bk$ at an angle $\theta$ from the $k_x$ axis, we find a continuous rotational symmetry:
\[
H(\theta) = \Gamma(\theta)H(0)\Gamma^*(\theta), \qquad \Gamma(\theta) = \text{diag} (R(\theta), R(\theta), I_2), 
\]
where $R(\theta)$ denotes the usual 2d rotation matrix by an angle $\theta$. This unitary equivalence shows that the spectrum of $\hat H^\beta(\bk)$ is invariant with respect to rotations of $\bk$. $\hat H^\beta(\bk)$ also obeys a parity symmetry:
\[
\Gamma_p\Gamma(\theta)^* H(\theta) \Gamma(\theta)\Gamma_p \equiv \Gamma_p(\theta)^*H(\theta)\Gamma_p(\theta) = -H(\theta), \qquad \Gamma_p = \text{diag}(-1, 1, 1, -1, 1, 1)
\]
so that the spectrum of $\hat H^\beta(\bk)$ is symmetric about 0. Finally, denoting $H(\theta, \om)$ to be $H(\theta)$ with a specific $\om$ value we find that:
\[
\Gamma_p(\theta)^*\Gamma_\om  H(\theta, \om) \Gamma_\om \Gamma_p(\theta) \equiv \Gamma_\om (\theta)^*H(\theta, \om)\Gamma_\om (\theta) = H(\theta, -\om), \ \Gamma_\om = \text{diag}(1, 1, -1, -1, 1, -1)
\]
so that the spectrum of $\hat H^\beta(\bk)$ is invariant with respect to a change in sign of $\om$. In particular we will it find useful that if $\psi$ is an eigenvector of $H(\theta, \om)$ then $\Gamma_\om(\theta)\psi$ is an eigenvector of $H(\theta, -\om)$ with the same eigenvalue.

\subsection{Chern numbers}
Integrals of Berry curvature of the eigenvectors corresponding to the bands $\tilde\omega_{\pm1}(k), \tilde \omega_{\pm2}(k)$ are $\mC_{\pm 1} = \mp \text{sgn}(\om_h)$, $\mC_{\pm 2} = \pm \text{sgn}(\om_h)$, $h \in \{N, S\}$ \cite{pakniyat2022chern, serra2025influence}. However, considering that the eigenvectors depend explicitly on $\bk$ as $|\bk| \to \infty$, these do not represent \textit{bona fide} Chern numbers, as the eigenvectors cannot continuously be mapped to a closed manifold. 

We again define the BDIs $\fC_{1, 2}$ for a transition between the phases $S$ ($\om < 0$) and $N$ ($\om >0$) (which we still must verify are well-defined Chern numbers) \cite{pakniyat2022chern}:
\[
\fC_1 = \sum_{j < 1}\mC_j^S -\sum_{j < 1} \mC_j^N = 0
\]
\[
\fC_2 = \sum_{j < 2}\mC_j^S -\sum_{j < 2} \mC_j^N = 2.
\]
From the explicit eigenvectors calculated by Serra and Silveirinha \cite{serra2025influence} we can deduce
\[
\lim_{k \to \infty} \psi_1(\theta = 0) = \frac 1{\sqrt{2}} \left(1, 0, 0, 0, 0, 1\right)^T \qquad \lim_{k \to \infty} \psi_2(\theta = 0) = \frac 1{\sqrt{2}} \left( 0, 0, 0, -1, -1, 0\right)^T.
\]
We can apply the unitary operator $\Gamma_\om(0)$ as derived above to show that $\lim_{k \to \infty}\psi_1^N(\theta =0) = -\lim_{k \to \infty} \psi_1^S(\theta =0)$ and $\lim_{k\to \infty}\psi_2^N(\theta=0) = \lim_{k \to \infty}\psi_2^S(\theta =0)$. Therefore we deduce that \eqref{eq:gluing} indeed holds for $\theta = 0$ and unitary transformation by $\Gamma(\theta)$ shows that \eqref{eq:gluing} holds for all $\theta \in \Sm^1$ so that $\fC_1, \fC_2$ are indeed well-defined Chern numbers for the hydrodynamic model without any need for regularization as $k \to \infty$.

\subsection{BEC in hydrodynamic model}

\begin{figure}[t!]
    \centering
    \includegraphics[width = \textwidth]{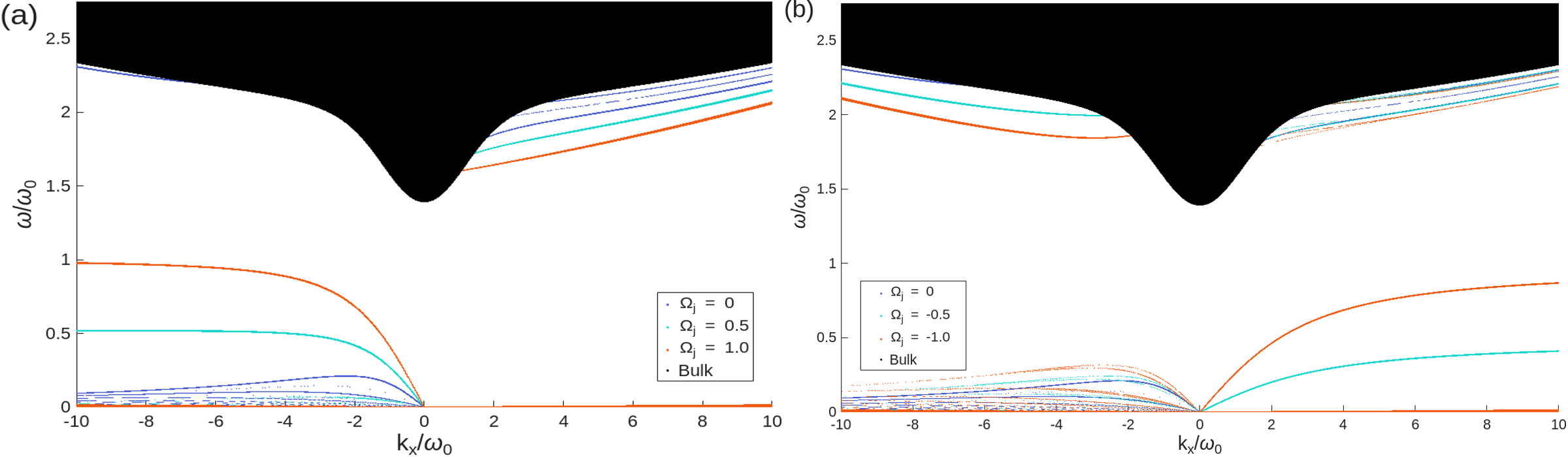}
\caption{Numerically calculated spectrum of $H^\beta_I(k_x)$ for a linear transition between $\om^S = -1$ THz and $\om^N = 1$ THz with a jump at $y = 0$ as in Figures \ref{fig:BEC_jumps} and \ref{fig:BEC_jumps_neg}. (a) shows positive values of $\om_j$ and (b) negative values of $\om_j$.  $\beta = 0.1$, $\op = 2$ THz, $\omega_0 = 1$ THz.}\label{fig:hydro}
\end{figure}

The BEC for the hydrodynamic model is modified in a similar manner as in the local case with a few important differences. First, from \eqref{eq:bulk_hydro_1} and \eqref{eq:bulk_hydro_2} we see that both positive bulk bands $\tilde \omega_{1, 2} \to \infty$ as $k \to \infty$, therefore the upper band gap is eliminated. We similarly term $\min\{\tilde \omega_1^N(0), \tilde \omega_1^S(0)\} = E_1$ so that the only band gap is now $(0, E_1)$. We find that $\omega_e^{(2)}(k_x) \to \infty$ as $k_x \to \infty$ so that $\omega_e^{(2)}$ can no longer contribute to the spectral flow in the lower band gap. Therefore we find that if $E \in (0, E_1)$ and $g$ is a connected interval $g \in (0, E_1) \backslash (\varepsilon_{1R} \cup \varepsilon_{1L})$ which contains $E$ then: 
    \begin{equation}\label{eq:BEC_hydro}
    2\pi \sigma_I(g) = \mJ_{1R}(E) -\mJ_{1L}(E)
    \end{equation}
where $\mJ_{1j}(E)$ and $\varepsilon_{1j}$, $j \in \{R, L\}$ are defined identically as in Section \ref{sec:cold_plasma_model}. Now although \eqref{eq:BEC_hydro} is identical to \eqref{eq:BEC_1_simp} the relation holds without any simplifying assumptions on $\om_j$. This result is demonstrated with simple linear transitions with a jump at $y = 0$ as in the previous sections in Figure \ref{fig:hydro}. Interestingly, in contrast to \cite{buddhiraju2020absence, hassani2019truly}, we find that introduction of a non-dissipative non-local term in the transition between two oppositely-biased plasmas still produces an asymptotically flat branch in the band gap as $k\to \infty$, although these works consider a slightly different case in which the transition is between a biased plasma and a topologically trivial transparent or opaque conductor.

\section{Discussion}\label{sec:discuss}

In the above we have defined and numerically demonstrated an anomalous BEC in TM wave propagation in a magnetically-biased continuum photonic system with finite width interface across which the cyclotron frequency $\om(y)$ varies. The number of edge modes is not only dependent on the well-defined Chern invariants $(\fC_1, \fC_2) = (0, 2)$ but also the number and amplitude of discontinuities in $\om(y)$. In the case that $\om(y)$ continuously transitions between two bulk values of $\om_S, \om_N$ in the region $|y|< 1$ we find that the BEC holds in both local and hydrodynamic models. When discontinuities in $\om(y)$ are introduced, we find at high wavenumbers that three edge modes ($\omega_e^{(1, 2, 3)}(k_x)$), whose eigenvectors are concentrated around the discontinuity, appear with flat asymptotic dispersions which may alter the spectral flow in the unregularized local and hydrodynamic models. Introduction of regularizations in the local model alter $\hat\omega_e^{(1, 2, 3)}$. In the BDI regularization $\hat\omega_e^{(1, 2, 3)}$ are altered in such a way that the BEC in fact holds regardless of any discontinuities in $\om(y)$. This picture is considerably more complicated in the SC regularization, in which the BEC holds for the upper band gap but the lower band gap is effectively altered in a way that is sensitively dependent on the size of the discontinuity and completely eliminated when $\om(y)$ is continuous.

Interestingly, the existence of $\omega_e^{(1, 2, 3)}(k_x)$ depends only on the existence of a jump in $\om(y)$ and not on any (bulk) topological transition. This suggests that a non-zero spectral flow may exist in a portion of the band gaps of the system whenever a sharp transition in $\om$ exists, even if no topological transition occurs. This behavior is shown in Figure \ref{fig:piecewise} for both local and hydrodynamic models with a sharp transition in $\om$ in which $\om_N, \om_S> 0$. Although these edge states are not ``topological" in the sense that they are not characterized by the bulk topological invariants $\fC_\ell$, they demonstrate the same robust asymmetric transport properties quantized by $\sigma_I(g)$ in an energy range $g$ determined by $\om_-$ and $\om_+$ via \eqref{eq:asym_freq}.

Our findings also show that considering a more general transition in $\om(y)$ across a thick interface $|y| < 1$ produces a more complicated spectral picture of photonic continua than that of a purely piecewise sharp transition in $\om$. One might expect for a sharp transition to be a valid approximation if the transition region is much smaller than the excitation wavelength. THz frequencies correspond to wavelengths $\lambda \approx 100\mu m$ ($c/\omega_0 = 300 \mu$m for Figures \ref{fig:cont_BEC}-\ref{fig:piecewise}) so for precisely fabricated transitions between two separate media this may be an appropriate assumption, while variations in magnetic field bias may be wider. In any case, even in magnetic sources designed to provide a constant magnetic field, some non-zero variation in field strength will occur \cite{martin2023magnetic} which we have shown will manifest in additional spectral bands which were not analyzed before.

Whether or not the transition is modeled as sharp or continuous, we have shown that the edge states which appear and energy regions which support a non-zero spectral flow depend sensitively on discontinuities in $\om(y)$ and we precisely quantized these anomalies in the BEC. While we have limited our analysis only to transitions between two magnetically biased photonic continua which are well-approximated by the cold plasma model, previous results suggest that an even richer area  of phenomena may exist in continuous transitions between more general continuum photonic models.

\begin{figure}[t!]
    \centering
    \includegraphics[width = \textwidth]{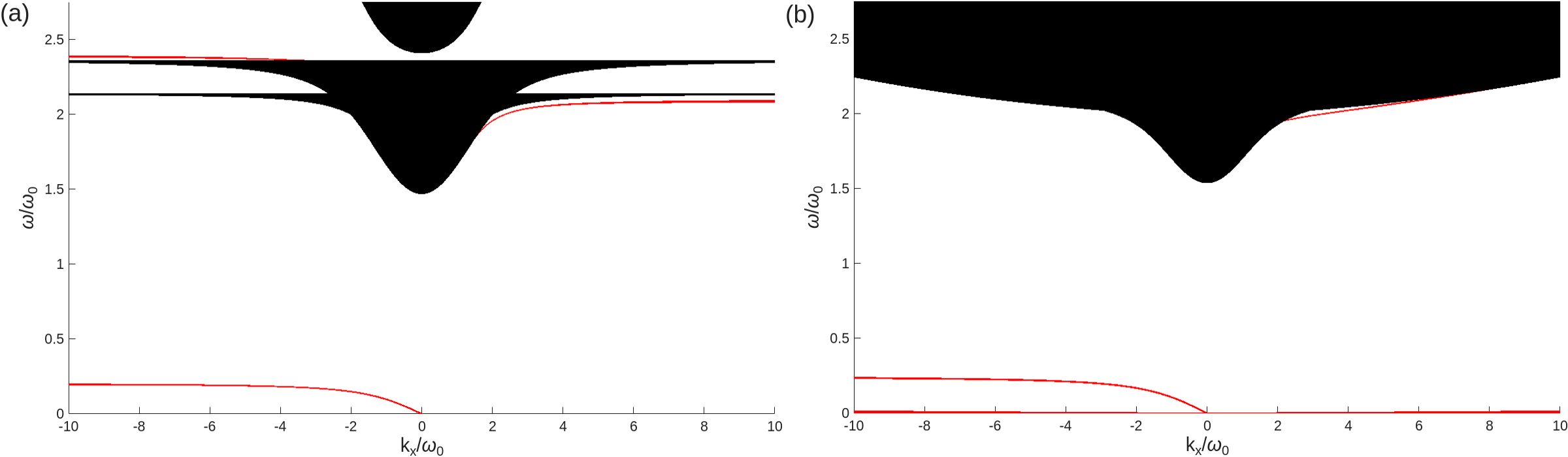}
\caption{Spectrum for a piecewise constant $\om(y)$ which does not undergo a topological phase transition for both the local model (a) and the hydrodynamic model (b). Areas of black indicate bulk spectrum and areas of red edge spectrum. For the local model $\om_S = 0.75$ THz and $\om_N = 1.25$ THz and for the hydrodynamic model $\om_S = 0.3$ THz and $\om_N = 0.8$ THz. A spectral flow of -1 appears in the lower band gap in the region $(0, \hat\omega_e^{(1)})$ for both models and in the region $(E_{uh}, \hat\omega_e^{(3)})$ for the local model. $\op = 2$ THz, $\omega_0 = 1$ THz, and $\beta = 0.1$.}\label{fig:piecewise}
\end{figure}

\section{Acknowledgements}
The authors acknowledge Jeremy Hoskins for invaluable advice and input regarding numerical techniques. This work was funded in part by NSF grant DMS-230641 and ONR
grant N00014-26-1-2017.

\section{Data Availability}\label{sec:data}
The code needed to reproduce all figures in this work is available at DOI: 10.5281/zenodo.20058114.

\appendix

\section{Interface current observable $\sigma_I$}\label{sec:sigma_I}

In this section we give a precise definition of the interface current observable $\sigma_I$. We assume as in Section \ref{sec:background} that $H_I(x, y, D_x, D_y)$ is a self-adjoint operator on an appropriate Hilbert space $\mH$ which satisfies $H_I(x, y, D_x, D_y) = H^N(D_x, D_y)$ for $y \ge 1$ and $H_I(x, y, D_x, D_y) = H^S(D_x, D_y)$ for $y \le -1$. We further assume that $H^N, H^S$ have a common global spectral gap $(E_0, E_1)$. See Appendix \ref{sec:5x5_analysis} for the definition of $\mH$ for the cold plasma model.

Let $P=P(x)\in C^\infty(\Rm)$ be a function that depends only on the spatial coordinate $x$ with  $P(x)=0$ for $x<x_0-\delta$ and $P(x)=1$ for $x>x_0+\delta$ for some $x_0\in\Rm$ and $\delta>0$. The function $P(x)$ should be interpreted as the observable quantifying the field density in the (right) half-space $x\geq x_0$. Then $i[H,P]$ with $[A,B]=AB-BA$ the standard commutator, may naturally be interpreted as a current operator modeling transfer of field density  per unit time across the (thick) vertical line where $P$ transitions from $0$ to $1$.

Since $H^N, H^S$ are insulators in the region $(E_0, E_1)$, any excitation generated in such an interval will be confined to the vicinity of the interface $y\approx0$ separating $H^N$ and $H^S$. We define $0\leq\varphi\in C^\infty(\Rm)$ as a monotonic function such that $\varphi(E)=0$ for $E\leq E_0$ and $\varphi(E)=1$ for $E\geq E_1$. Thus $\varphi'(H_I)\ge 0$ defines a density of states of modes that cannot propagate into the $N$ and $S$ bulks. The expectation value of the current observable $i[H_I,P]$ for a density of states $\varphi'(H_I)$ is then defined as:
\begin{equation}\label{eq:sigmaI}
  \sigma_I [H_I] = {\rm Tr} \ i[H_I,P] \varphi'(H_I)
\end{equation}
assuming that $i[H_I,P] \varphi'(H_I)$ is a trace-class operator. Here, ${\rm Tr}$ is the standard trace on the Hilbert space $\mH$ where $H_I$ is defined. We refer to $\sigma_I$ as an interface (current) observable. It is the physical object describing asymmetric transport along an interface, which here is modeled as the $x-$axis. This invariant can be interpreted as the spectral flow across the insulating gap $(E_0, E_1)$ \cite[\S3]{bal2024continuous}, which is how we evaluate it from numerical spectral calculations.

Under the appropriate assumptions $\sigma_I[H_I]$ is in fact independent of the particular choice $P(x)$ and $\varphi(E)$ as long as they satisfy the above criteria \cite{bal2022topological}. In particular we assume invariance in the $x-$direction so that the choice of $x_0, \delta$ is immaterial. Therefore only the support of $\varphi'(E)$, which is the interval $g$ in the main text, matters when computing $\sigma_I$, allowing us to define $\sigma_I(g)$ which is defined by \eqref{eq:sigmaI} with any function $\varphi$ satisfying the above conditions and for which ${\rm supp}(\varphi') = g$. 

\section{Numerical methods}\label{sec:numerics}
Code for implementing the following method is freely available for download at the DOI provided in Section \ref{sec:data} \cite{frazier2026code}. Permuting the basis of $\psi$ in Section \ref{sec:cold_plasma_model} $(v_x, v_y, E_x, E_y, B_z) \to (E_x, B_z, v_x, E_y, v_y)$ from \eqref{eq:HTMTE} we get:
\[
H_I = \begin{pmatrix}
-D_y \sigma_1 & B\\
B^* & C
\end{pmatrix}
\]
for:
\[
\sigma_1 = \begin{pmatrix}
    0 & 1\\
    1 & 0
\end{pmatrix}, \quad
B = \begin{pmatrix}
    i\op& 0 & 0 \\
    0 & k_x & 0
\end{pmatrix}, \quad
C = \begin{pmatrix}
    0 & 0 & i\om(y)\\
    0 & 0 & i\op\\
    -i\om(y) & -i\op & 0
\end{pmatrix}.
\]
Separating $\psi = (\psi_1, \psi_2)^T$ and rearranging we get:
\[
\begin{pmatrix}
A & B\\
B^* & D
\end{pmatrix}
\begin{pmatrix}
    \psi_1\\
    \psi_2
\end{pmatrix} = -i\begin{pmatrix} 
\sigma_1 & 0 \\
0 & 0
\end{pmatrix} \partial_y\begin{pmatrix}\psi_1\\\psi_2\end{pmatrix}
\]
for $A = -\omega I_2, \; D = C-\omega I_3$. Assuming that $D$ is invertible and taking the Schur complement in $D$ gives:
\begin{equation}\label{eq:ode}
i\sigma_1(A-BD^{-1}B^*)\psi_1=: V(k_x, \omega)\psi_1(y)= \partial_y \psi_1(y).
\end{equation}
An identical procedure can be used to define $V(k_x, \omega)$ in this way for the Hamiltonian $H_I^\beta$ from \eqref{eq:H_TM_hydro}. This is now an ODE which can be solved using ubiquitous Runge-Kutta or other standard numerical ODE solving methods. Assuming that for large enough $|y|$, the coefficients of $V(k_x, \omega)$ are constant and given by the constant-coefficient operator $V^{N}(k_x, \omega), V^S(k_x, \omega)$ for $y\geq 1$ and $y \leq -1$ respectively, then for any fixed $(k_x, \omega)$ we can determine the bulk modes by simply diagonalizing $V^{N}, V^S$. In general, however, an explicit solution for $(\psi(y),\omega)$ is not feasible analytically for the ODE \eqref{eq:ode}. A numerical ODE solver is used to solve \eqref{eq:ode} for $\psi(y=0)$ by using the eigenvectors of $V^N$ and $V^S$ as initial conditions at $y = \pm 1$ respectively and numerically solving for $\psi(y)$ for $-1< y < 0$ and $0 < y < 1$ respectively. Exponentially increasing eigenvectors from the bulk spectrum of $V^N$ and $V^S$ (eigenvectors whose eigenvalues have real parts $>0$ for $V^N$ and $< 0$ for $V^S$) are eliminated as non-physical. If the subspaces of valid (non-exponentially increasing) eigenvectors of $V^N, V^S$ intersect non-trivially once evaluated at $y = \pm 0$ then there exists a valid solution of \eqref{eq:ode}. More precisely, orthogonal projectors $\Pi^N, \Pi^S$ onto the subspaces of valid eigenvectors of $V^N, V^S$ evaluated at $y = 0$ are formed, and a $(k_x, \omega)$ pair is accepted if the largest eigenvalue of $(\Pi^N\Pi^S)$ is within a small tolerance ($<10^{-5}$ for all figures in the main text) of 1. When compared with widely-used finite difference methods \cite{fu2021topological, bal2024topological, bal2022multiscale} this method has no need for periodization of the domain and therefore does not require the heuristic elimination of modes \cite{bal2024topological, bal2022multiscale} or alternatively the necessity of modeling two equal but opposite transitions \cite{fu2021topological, zhang2011spontaneous}. This method shares these advantages with two previously-proposed approaches \cite{colbrook2019compute, colbrook2023computing} but may also be applied to continuum models as well as periodic lattice models. Comparison with finite-difference methods (not displayed here) does however show that the two methods agree using a fine enough mesh with finite differences. 

There are two issues which arise using this method that are worth discussion. First, we note that some of the constraints are purely algebraic (represented by the first 3 rows of \eqref{eq:HTMTE} or $D$ in \eqref{eq:ode}), while the remaining contain derivatives. Strictly speaking our system then represents a differential algebraic equation, whose numerical solutions have convergence and stability properties which in general require much more analysis beyond traditional results for numerical ODE solvers (see \cite{kunkel2006differential} for an overview of differential algebraic equations and their numerical solutions). However in our case the fact that $D$ is non-singular whenever $\omega\ne \pm\sqrt{\op^2+\om^2}, 0$, which in particular is termed a ``strangeness free" differential algebraic equation with no nilpotent parts, allows us to circumvent these additional difficulties. The issue which remains is that inversion of $D$ is ill-conditioned whenever $\omega$ is close to $\sqrt{\op^2+\om^2}$ or 0. The resulting error will depend in particular on which eigenvectors of $(\Pi^N\Pi^S)$ coincide with directions which are close to the kernel of $D$ when $\omega = \sqrt{\op^2+\om^2}, 0$ and therefore it is hard to predict which branches of spectrum are most affected by this error. However we can see it manifested in the thickness of calculated spectral branches in the edge spectrum in Figures \ref{fig:BEC_jumps}- \ref{fig:BEC_2_jumps}. In particular the edge modes $\hat \omega_e^{(2, 3)}(k_x)$ become thicker as $|k_x|$ becomes large. The bulk modes are calculated analytically and therefore are unaffected by the fact that $D$ is singular for certain $\omega$ values in the bulk spectrum.

The second issue stems from the fact that this is a grid-based method which calculates whether a particular $(k_x, \omega)$ point approximates a point in the spectrum of $H_I(k_x)$ rather than producing a discrete approximation of the spectrum of $H_I(k_x)$ itself. In areas where discrete branches of spectrum are densely packed, as we see near $\omega = 0$ and $\omega\in (\op, \sqrt{\op^2+\om^2})$ our method would need a grid spacing finer than the density of states in order to resolve each branch, as opposed to finite difference approximations, which instead approximate a fixed, finite number of eigenvalues for each $k_x$ value. Compounding this issue is the fact that inversion of $D$ is ill-conditioned near these areas of high spectral density in our case so that our numerical method performs poorly in the regions near $\omega = 0, (\op, \sqrt{\op^2+\om^2})$. However, exactly because $0, (\op, \sqrt{\op^2+\om^2})$ are regions of bulk spectrum, these areas of poor performance do not affect our analysis of the spectral flow and BEC except for the SC regularization, which deforms these spectral bands into the lower band gap.

\section{Analysis of Local Model}\label{sec:5x5_analysis}
In this section we analytically derive the existence of the edge states $\omega_e^{(j)}(k_x)$ for high wavenumbers and prove the results \eqref{eq:BEC_1},\eqref{eq:BEC_2}, and \eqref{eq:BEC_1_simp} in limited cases.

First, in order to ensure $H_I(k_x)$ is self-adjoint we must precisely define its domain. Assuming our underlying Hilbert space to be $L^2(\Rm;\Cm^5)$ we can deduce the domain simply by enforcing $\langle\phi, H_I(k_x)\psi\rangle =\langle H_I(k_x)\phi, \psi\rangle\; \forall \phi, \psi \in \mathcal{D}(H_I(k_x)) $, from which we deduce that $E_x(y), B_z(y) \in H^1(\Rm)$ and $v_x(y), v_y(y), E_y(y)\in L^2(\Rm)$. From Morrey's inequality \cite{hunter2001applied} we get that $E_x(y), B_z(y) \in C_0(\Rm)$ while the remaining coordinates may in general be discontinuous. This will be important in our analysis below.

\paragraph{Piecewise $\om(y)$.}
We first analyze the case where $\om(y)$ is piecewise constant, so that $\om(y) = \om_N$ for $y\ge 0$ and $\om(y) = \om_S$ for $y< 0$. We can also express this as $\om(y) = \om_a +\text{sgn}(y)\om_j$ for the average and half-jump values:
\[
\om_a = \frac{\om_S+\om_N}{2}\qquad \om_j = \frac{\om_N-\om_S}{2}.
\]
For edge states we assume solutions of the form $\psi(y) = e^{\kappa_S y}\psi_0^S$ for $y< 0$ and $\psi(y) = e^{-\kappa_Ny}\psi_0^N$ for $y\ge 0$, $\kappa_N,\kappa_S > 0$. With solutions of this kind we can substitute $-D_y \to i\kappa$ in \eqref{eq:HTMTE} along with $D_x \to k_x$. We then solve for $\kappa$ in the equation $\det(H_I(\kappa)-E)=0$ which gives:
\begin{equation}\label{eq:kappa}
\kappa^2 = -\frac{E^4-E^2(\omega_h^2+k_x^2)+k_x^2\omega_{uh}^2 + \op^4}{E^2-\omega_{uh}^2}=k_x^2-\frac{E^4-\omega_h^2E^2+\op^2}{E^2-\omega_{uh}^2},
\end{equation}
defining $\omega_{uh}^2 =\op^2+\om^2$ and $\omega_h^2 = 2\op^2+\om^2$. Therefore we have an equation for $\kappa_N, \kappa_S$. In order to satisfy $H_I\psi = E\psi$ at $y = 0$ we introduce the notation $[u] = u(0^+)-u(0^-)$ and apply $[H_I\psi] = E[\psi]$. Recall from above that in order for $H_I(k_x)$ to be self-adjoint $E_x(y), B_z(y)$ must be continuous while $v_x(y), v_y(y), E_y(y)$ may in general have discontinuities. Therefore we obtain:
\[
i[\om v_y] = E[v_x], \quad -i[\om v_x] + i\op[E_y] = E[v_y], \quad -i\op[v_x] +i[\kappa]B_z(0) = 0
\]
\[
-i\op [v_y]  = E[E_y], \quad i[\kappa]E_x(0) + k_x[E_y] = 0.
\]
In general solutions for $E$ involve substituting \eqref{eq:kappa} into the above system and solving for $E$ and are quite intractable. Conveniently in \eqref{eq:ode} we have eliminated all coordinates except $\psi_1 = (E_x, B_z)$, which must be continuous, and since the constraints on $v_x, v_y, E_y$ are purely algebraic we can obtain a unique eigenvector by solving \eqref{eq:ode} as long as $E\ne 0, \pm \sqrt{\om^2 + \op^2}$ (values for which $D$ is not invertible). Calculating $V(k_x, E)$ from \eqref{eq:ode} explicitly gives:
\begin{equation*}%\label{eq:edge_ode}
V(k_x, E) = \frac 1{E(E^2-\op^2-\om^2)}
\begin{pmatrix}
    \op^2\om k_x & -i(E^4-E^2(k_x^2+\op^2+\om^2) + k_x^2 \om^2)\\
    -i(E^4-E^2(2\op^2+\om^2) + \op^4) & -\op^2\om k_x
\end{pmatrix}.
\end{equation*}
Notice from \eqref{eq:ode} that the eigenvalues of $V(k_x, E)$ should be exactly $\kappa$ as defined in \eqref{eq:kappa}, which is readily confirmed by straightforward diagonalization of $V$. Our goal now is to solve $[V(k_x, E)\psi_1] = [\partial_y \psi_1]$ for the exponentially decaying solutions posited above, remembering that $\psi_1 = (E_x, B_z)$ must be continuous at $y = 0$. 

First consider the symmetric case in which $\om(y) = \text{sgn}(y)\om_j$ and note that the off-diagonal terms do not depend on the sign of $\om$. We can deduce edge state dispersions in this case:
\[E^4-E^2(k_x^2+\op^2+\om_j^2) + k_x^2 \om_j^2 = 0 \Rightarrow \]
\begin{equation*}\label{eq:edge_dis_1}
E_\pm^2(k_x) = \frac 12 \left(k_x^2 + \omega_{uh}^2 \pm \sqrt{(k_x^2 + \omega_{uh}^2)^2 -4 \om_j^2k_x^2}\right) = \frac 12 \left(k_x^2 + \omega_{uh}^2 \pm \sqrt{(k_x^2 + \op^2-\om_j^2)^2 +4 \op^2\om_j^2}\right)
\end{equation*}
for edge state solutions:
\[
\psi_1 = e^{\kappa y}\begin{pmatrix}
    0\\
    B_{0}
\end{pmatrix}, \qquad \kappa = -\frac{ k_x \op^2 \om_j}{E_\pm(k_x)(E_\pm(k_x)^2-\op^2-\om_j^2)}.
\]
Note that $E_+ > \sqrt{\op^2+\om_j^2}$ and $E_- < \sqrt{\op^2+\om_j^2}$ so that if $\om_j > 0$ the $E_+$ branch can only exist when $k_x > 0$ and the $E_-$ branch can only exist when $k_x < 0$ to ensure $\kappa < 0$ (and vice versa if $\om_j < 0$). Similarly we have edge state dispersions:
\[
E^4-E^2(2\op^2+\om_j^2) + \op^4 = 0 \Rightarrow E = \tilde\omega_1(0), \tilde\omega_2(0)
\]
for solutions:
\[
\psi_1 = e^{\kappa y}\begin{pmatrix}
    E_0\\
    0
\end{pmatrix}, \qquad \kappa = \frac{ k_x \op^2 \om_j}{\tilde\omega_{1, 2}(0)(\tilde\omega_{1, 2}^2(0)-\op^2-\om_j^2)}
\]
where again $\tilde\omega_1(0) < \sqrt{\op^2+\om_j^2}$ and $\tilde\omega_2(0) > \sqrt{\op^2+\om_j^2}$ so that the $\tilde\omega_1$ branch is defined only for $k_x > 0$ and the $\tilde \omega_2$ branch is defined only for $k_x < 0$. Setting these expressions for $\kappa$ equal to \eqref{eq:kappa} and substituting the dispersion relations confirms that the two expressions for $\kappa$ agree. Therefore we have deduced four branches of edge spectrum $\omega_e^{(1)}(k_x\le 0) = E_-(k_x)$, $\omega_e^{(2)}(k_x\ge0) = \tilde \omega_1(0)$, $\omega_e^{(3)}(k_x\le 0) = \tilde\omega_2(0)$, and $\omega_e^{(4)}(k_x \ge 0) = E_+(k_x)$ in the symmetric case $\om(y) = \text{sgn}(y) \om_j$ for $\om_j > 0$. The signs of $k_x$ are reversed for $\om_j < 0$. Note that $\omega_e^{(4)}(k_x)$ always converges to $\infty$ as $|k_x| \to \infty$ so that it always contributes $\pm1$ to the spectral flow of the upper band gap (depending on the sign $\pm$ of $\om_j$). Meanwhile $\omega_e^{(1, 2, 3)}(k_x)$ all have finite limits as $|k_x| \to \infty$ which depend on $\om_j$.

Unfortunately the general piecewise problem $\om(y) = \om_a + \text{sgn}(y)\om_j$ does not present such tractable solutions. Instead notice that the general piecewise case is a bounded perturbation such that if $H_0(k_x)$ is the symmetric case $\om(y) = \text{sgn}(y)\om_j$ we obtain the general piecewise case as $H_I(k_x) = H_0(k_x) -i\text{sgn}(y)\text{diag}(\sigma_2, 0, 0, 0)\om_a$. Standard results (see e.g. Theorem 3.6 of \cite[\S IV.3]{kato1966perturbation}) show that the spectrum of a closed operator under bounded perturbations with which it commutes is continuous. Therefore since $E_\pm(k_x), \tilde \omega_{1,2}(0)$ are in the point spectrum of $H_0(k_x)$, for a value of $\om_a$ which is not too large we also have eigenvectors of $H_I(k_x)$, $(\omega_e^{(1)}(k_x), \omega_e^{(2)}(k_x), \omega_e^{(3)}(k_x), \omega_e^{(4)}(k_x))$ which are close to $(E_-(k_x), \tilde\omega_1(0), \tilde\omega_2(0), E_+(k_x))$ respectively. In particular since the norm of our bounded perturbation to $H_0(k_x)$ is $|\om_a|$ we have that the distance from the eigenvalues of $H_I(k_x)$ to the eigenvalues of $H_0(k_x)$ is less than or equal to $|\om_a|$. Thus we know that $\omega_e^{(1)}(k_x)$ exists if $|\om_a|< |\om_j|$, $\omega_e^{(2)}(k_x)$ exists when $|\om_a| < \sqrt{\op^2+\om_S^2} - \tilde\omega_1(0)$ and $\omega_e^{(3)}(k_x)$ exists when $|\om_a| < \tilde\omega_2(0)-E_{uh}$. Our numerical spectral calculations in Figures \ref{fig:BEC_jumps} and \ref{fig:BEC_jumps_neg} show that these branches in fact exist for any $\om(y)$ which changes sign.

We can however do even better and determine the asymptotic behavior of $\omega_e^{1, 2, 3}(k_x)$ whenever they exist. Notice from \eqref{eq:kappa} that as $\lim_{k_x \to \pm\infty}\kappa =|k_x|$ assuming that $E$ remains bounded. Solving $[V(k_x, E)\psi_1] = [\partial_y \psi_1]$ for exponentially decaying solutions for which $\psi_1$ is continuous, in the limit as $k_x \to \infty$ then gives (using the notation $\lim_{k_x\to \pm \infty}\omega_e^{(i)}(k_x) = \hat \omega_e^{(i)})$:
\[
(\hat \omega_e^{(i)})^6 -(\hat \omega_e^{(i)})^4(2\op^2+\om_N^2+\om_S^2) + (\hat \omega_e^{(i)})^2 (\op^2+\om_N^2)(\op^2+\om_S^2) -\op^4\om_j^2 = 0
\]
Explicit expressions for the zeros of this polynomial (cubic in $E^2$) are available and take the form:
\[
(\hat \omega_e^{(1)})^2 = \frac{2\op^2+\om_N^2+\om_S^2}{3} - \frac23\sqrt{q}\cos\left(\frac 13(\theta+\pi)\right)
\]
\[
(\hat \omega_e^{(2)})^2 = \frac{2\op^2+\om_N^2+\om_S^2}{3} - \frac23\sqrt{q}\cos\left(\frac 13(\theta-\pi)\right)
\]
\[
(\hat \omega_e^{(3)})^2 = \frac{2\op^2+\om_N^2+\om_S^2}{3} + \frac23\sqrt{q}\cos\left(\frac \theta3\right)
\]
\[
q = (\op^2+\om_N^2)^2 + (\op^2+\om_S^2)^2 -(\op^2+\om_N^2)(\op^2+\om_S^2), \quad \theta \in [0, \pi).
\]
With $\theta \in [0, \pi)$ we get that $\cos(\theta/3) \in (1/2, 1]$ and $\cos(1/3(\theta-\pi))\in [-1/2, 1/2)$. Assuming WLOG that $\om_S^2\le \om_N^2$ we get from the above expressions that $\hat\omega_e^{(3)}>\sqrt{\op^2+\om_S^2}$ and $\hat \omega_e^{(1)} \le \hat\omega_e^{(2)} \le \sqrt{\op^2+\om_N^2}$ so that $\omega_e^{(3)}$ can only contribute to the spectral flow of the upper band gap and $\omega_e^{(1,2)}$ can only contribute to the spectral flow of the lower band gap. Thus, counting the spectral flows contributed by $\omega_e^{(i)}(k_x)$, $i \in \{1, 2, 3, 4\}$ using $\omega_e^{(i)}(0)$ and $\hat \omega_e^{(i)}$ for the appropriate sign of $k_x$, we have proved the relations \eqref{eq:BEC_1} and \eqref{eq:BEC_2} for any case where $\om(y) = \om_a + \text{sgn}(y) \om_j$ with $|\om_a|$ not too large.

\paragraph{General $\om(y)$.}We now extend these results to the case where $\om(y)$ is a general piecewise continuous function with a discontinuity at $y_0 \in (-1,1)$. We assume that $\om(y_0^\pm) = \om_\pm$. The main result is Theorem \ref{thm:existence_edge}, which proves the existence of branches of spectrum $E_j(k_x)$, $j \in \{1, 2, 3\}$ which converge to $\omega_e^{(j)}(k_x)$ as $|k_x| \to \infty$, where $\omega_e^{(j)}(k_x)$ are defined identically to the last paragraph from \eqref{eq:asym_freq} with $\om_- = \om_S$ and $\om_+ = \om_N$. The remaining ingredient to proving the relations \eqref{eq:BEC_1}, \eqref{eq:BEC_2}, and \eqref{eq:BEC_1_simp} are the behavior of the branches $E_j(k_x)$ as $k_x \to 0$, allowing us to determine the spectral flow contributed by each branch as stated in the main results. This would involve explicit analysis of the spectrum of $H_I$ for low $|k_x|$ and therefore we leave this as a conjecture which is confirmed by numerical spectral calculations.

\begin{theorem}\label{thm:existence_edge}
Suppose that $\om(y)$ has a discontinuity at $-1< y_0< 1$ with $\om(y_0^\pm) = \om_\pm$ and $\om'(y) = 0$ in a small region $y \in [y_0-2\eta, y_0+2\eta]$. Then there exist solutions $(\psi_{j, k_x}(y), E_j(k_x))$, $j \in\{1, 2, 3\}$ to \eqref{eq:HTMTE} for which $\psi_{j,k_x}(y)$ decays exponentially away from $y = y_0$ and for which:
\[
|E_j(k_x) - \omega_e^{(j)}(k_x)| \le Ce^{-\alpha|k_x|}
\]
where $\omega_e^{(j)}(k_x)$ are the branches of edge spectrum for the piecewise constant problem with $\om_N = \om_+$, $\om_S = \om_-$ and $C, \alpha > 0$ are appropriately chosen constants which do not depend on $k_x$. $E_1, E_3$ are defined for $k_x \le 0$ and $E_2$ for $k_x \ge 0$ when $\om_+ > \om_-$ and vice versa for $\om_+ < \om_-$. If we assume instead that $\om(y)$ is locally Lipschitz in the region $[y_0-2\eta, y_0 + 2\eta]$ then $E_j(k_x) \to \hat \omega_e^{(j)}$ but the convergence to $\omega_e^{(j)}(k_x)$ is no longer exponential.
\end{theorem}

\proof
We assume WLOG that a discontinuity in $\om(y)$ exists at $y = 0$ and $\om(y)$ is continuous in some interval $[-2\eta, 2\eta]$ around 0. First we assume that $\om(y) = \om_+$ for $0 < y < 2\eta$ and $\om(y) = \om_-$ for $-2\eta < y< 0$. Denote the solutions $(\psi_j(k_x), \omega_{e}^{(j)}(k_x))$, $j \in \{1, 2, 3\}$ as eigenvectors and eigenvalues respectively of $H_I(k_x)$ with $\om(y) = \om_-$, $y< 0$; $\om(y) = \om_+$, $y \ge 0$, which we found in the previous section to be of the form:
\[
\psi_j(y) = \psi_{j0}(y) \gamma(y), \qquad \gamma(y) = \begin{cases}
    e^{-\kappa_Ny} & y\ge 0\\
    e^{\kappa_S y} & y < 0
\end{cases}
\]
where $\psi_{j0}(y) \in \Cm^5$ is constant in the first two coordinates and piecewise constant in the last three coordinates. Now introduce $\phi \in C^\infty_c(\Rm)$ such that $\sup_{y\in \Rm}\phi(y) \le 1$, $\phi(y) = 0$ for $|y| \ge 2$, and $\phi(y) = 1$ for $|y| \le 1$ and denote $\varphi_{j\eta} =\psi_j(y)\phi\left(\frac y\eta\right)$. Then since $\om(y) = \om_+$ or $\om(y) = \om_-$ on the support of $\varphi_{j\eta}$ we obtain:
\[
H_I(k_x)\varphi_{j\eta} = \omega_e^{(j)}(k_x)\varphi_{j\eta} + \frac i\eta\phi'\left(\frac y\eta\right)\gamma(y)
\big(0,0,\psi_{j0}^{(5)},0,\psi_{j0}^{(3)}\big)^T.
%\begin{pmatrix} 0\\0\\ \psi_{j0}^{(5)}\\0 \\ \psi_{j0}^{(3)}\end{pmatrix}
\]
Since $\phi\in C_c^\infty$, $\phi'$ is compactly supported and bounded, so that:
\[
||\left(H_I(k_x)-\omega_e^{(j)}(k_x)\right)\varphi_{j\eta}||_2^2 = \frac C\eta^2 \int_\Rm \phi'\left(\frac y\eta\right)^2 \gamma^2(y) dy \le C \frac{||\phi'||_\infty^2}{\eta^2}\left(\int_\eta^{2\eta}e^{-2\kappa_Ny}dy + \int_{-2\eta}^{-\eta}e^{2\kappa_Sy}dy\right).
\]
Assuming that $\omega_e^{(j)}(k_x)$ is bounded we see that for sufficiently large $|k_x|$, $\kappa\ge C|k_x|$ via the relation \eqref{eq:kappa}. Therefore we get that, for $|k_x|$ sufficiently large:
\[
||\left(H_I(k_x)-\omega_e^{(j)}(k_x)\right)\varphi_{j\eta}||_2 \le C_\eta e^{-\eta C|k_x|}
\]
We now drop the assumption that $\om(y)$ is constant on the intervals $[-2\eta, 0)$, $[0, 2\eta]$. Noting that (by assumption) $\lim_{y\to 0^-}\om(y) = \om_-$ we have that $\om(y)$ is Lipschitz in the intervals $[-2\eta, 0]$, $[0, 2\eta]$ so that when $y \in [0, 2\eta] $ $|\om(y) -\om_+|\le K_N|y|$ and when $y \in [-2\eta, 0)$, $|\om(y)-\om_-| \le K_S|y|$ for non-negative Lipschitz constants $K_S, K_N$. Using the same estimate $\varphi_{j\eta}$ we get that:
\[
H_I(k_x)\varphi_{j\eta} = \omega_e^{(j)}(k_x)\varphi_{j\eta} +  i\gamma(y)\begin{pmatrix}
    (\om(y)-\om_+)\phi\left(\frac y\eta\right)\psi_{j0}^{(2)}\\ -(\om(y)-\om_+)\phi\left(\frac y\eta\right)\psi_{j0}^{(1)}\\  \frac 1\eta\phi'\left(\frac y\eta\right)\psi_{j0}^{(5)}\\0\\ \frac1\eta\phi'\left(\frac y\eta\right)\psi_{j0}^{(3)}
\end{pmatrix}, \;\;\; y\ge0\]
\[
H_I(k_x)\varphi_{j\eta} = \omega_e^{(j)}(k_x)\varphi_{j\eta} +  i\gamma(y)\begin{pmatrix}
    (\om(y)-\om_-)\phi\left(\frac y\eta\right)\psi_{j0}^{(2)}\\ -(\om(y)-\om_-)\phi\left(\frac y\eta\right)\psi_{j0}^{(1)}\\  \frac 1\eta\phi'\left(\frac y\eta\right)\psi_{j0}^{(5)}\\0\\ \frac1\eta\phi'\left(\frac y\eta\right)\psi_{j0}^{(3)}
\end{pmatrix}, \;\;\; y<0
\]
Therefore we get by the triangle inequality:
\[
||H_I(k_x)\varphi_{j\eta}-\omega_e^{(j)}(k_x)||_2^2 \le \frac {C_1} {\eta^2} \int_\Rm \phi'\left(\frac y\eta\right)^2 \gamma^2(y) dy \;+
\]
\[C_2\left( \int_{-2\eta}^0 e^{2\kappa_Sy}(\om(y)-\om_-)^2\phi^2\left(\frac y\eta\right)dy + \int_0^{2\eta}e^{-2\kappa_Ny}(\om(y)-\om_+)^2\phi^2\left(\frac y\eta\right)dy\right).
\]
The first integral we have already estimated above. For the second integral we can use the Lipschitz property of $\om(y)$ to bound:
\[
 \int_{-2\eta}^0 e^{2\kappa_Sy}(\om(y)-\om_-)^2\phi^2\left(\frac y\eta\right)dy\le K_S ||\phi^2||_\infty \int_{-2\eta}^0y^2e^{2\kappa_Sy}dy = \frac{C_S}{\kappa_S^3}\left(2-e^{-2\eta\kappa_S}(4\kappa_S^2\eta^2+4\kappa_S\eta + 2)\right).
\]
Again for sufficiently large $k_x$ we get $\kappa \ge C|k_x|$ so that:
\[
\int_{-2\eta}^0 e^{2\kappa_Sy}(\om(y)-\om_-)^2\phi^2\left(\frac y\eta\right)dy \le \frac{C_S}{|k_x|^3}\left(2-e^{-2\eta C|k_x|}(4\kappa_S^2\eta^2+4\kappa_S\eta + 2)\right).
\]
An identical estimate can be made for the third integral. Therefore we get that as $|k_x|\to \infty$:
\[
||H_I(k_x)\varphi_{j\eta}-\omega_e^{(j)}(k_x)||_2^2 = O\left(\frac 1{|k_x|^3}\right).
\]
Then by standard estimates (see e.g. \cite{kato1966perturbation}) $\sup_{E\in\sigma(H_I(k_x))} |E-\omega_e^{(j)}(k_x)|^{-1} = ||(H_I(k_x)\varphi_{j\eta}-\omega_e^{(j)}(k_x))^{-1}||_2 \ge C_\eta e^{\eta C|k_x|}$ if we assume that $\om(y)$ is constant in the intervals $[-2\eta, 0)$, $[0, 2\eta]$. Similarly if $\om(y)$ is merely Lipschitz continuous on the same intervals we have that $||(H_I(k_x)\varphi_{j\eta}-\omega_e^{(j)}(k_x))^{-1}||_2 \ge C |k_x|^3$. $E_j(k_x)$ is therefore defined as $\arg\sup_{E\in \sigma(H_I(k_x))}|E-\omega_e^{(j)}(k_x)|$. The full interface operator $H_I$ (without taking the Fourier transform in $x$) is self adjoint on an appropriate domain (described above) so that the branches $\omega_e^{(j)}(k_x)$ are analytic in $k_x$ and thus $||(H_I(k_x)-\omega_e^{(j)}(k_x))\varphi_{j\eta}||$ is also bounded for values of $k_x$ near 0. Therefore the preceding inequalities in fact hold for all values of $k_x$ for possibly larger constants. We have from above that $E_j(k_x) \to \hat\omega_e^{(j)}$ as $|k_x|\to \infty$ (for the appropriate sign of $k_x$), which completes the proof. \endproof

\bibliographystyle{unsrt}
\bibliography{Refs}

\end{document}